\documentclass[11pt]{article} 
\usepackage{algorithmic}
\usepackage{algorithm}
\usepackage{times} 
\usepackage{amsfonts}
\usepackage{amsmath} 
\usepackage{latexsym} 
\usepackage{epsfig}
\usepackage{latexsym} 
\usepackage{verbatim}
\begin{document}
 
%%%%%%%%%%%%%%%%%%%%%%%%%%%%%%%%% 
%  Enviroments 
\newtheorem{theorem}{Theorem} 
\newtheorem{corollary}[theorem]{Corollary}
\newtheorem{prop}[theorem]{Proposition} 
\newtheorem{problem}[theorem]{Problem}
\newtheorem{lemma}[theorem]{Lemma} 
\newtheorem{remark}[theorem]{Remark}
\newtheorem{observation}[theorem]{Observation}
\newtheorem{defin}[theorem]{Definition} 
\newtheorem{example}[theorem]{Example}
\newtheorem{conj}{Conjecture} 
\newenvironment{proof}{{\bf Proof:}}{\hfill$\Box$} 
\newcommand{\PR}{\noindent {\bf Proof:\ }} % beginning of a
% proof
\def\EPR{\hfill $\Box$\linebreak\vskip.5mm} % end of a proof
 
%%%%%%%%%%%%%%%%%%%%%%%%%%%%%%%%%%% 
%  Operators 
\def\Pol{{\sf Pol}} 
\def\mPol{{\sf MPol}} 
\def\Polo{{\sf Pol}_1\;} 
\def\PPol{{\sf pPol\;}} 
\def\Inv{{\sf Inv}}
\def\mInv{{\sf MInv}} 
\def\Clo{{\sf Clo}\;} 
\def\Con{{\sf Con}} 
\def\concom{{\sf Concom}\;} 
\def\End{{\sf End}\;}
\def\Sub{{\sf Sub}\;} 
\def\Im{{\sf Im}} 
\def\Ker{{\sf Ker}}
\def\H{{\sf H}}
\def\S{{\sf S}} 
\def\D{{\sf P}} 
\def\I{{\sf I}} 
\def\Var{{\sf var}} 
\def\PVar{{\sf pvar}} 
\def\fin#1{{#1}_{\rm fin}}
\def\P{{\sf P}} 
\def\Pfin{{\sf P_{\rm fin}}} 
\def\R{{\rm R}} 
\def\F{{\rm F}} 
\def\Term{{\sf Term}}
\def\var#1{{\sf var}(#1)} 
\def\Sg#1{{\sf Sg}(#1)} 
\def\Cg#1{{\sf Cg}(#1)} 
\def\tol{{\sf tol}} 
\def\rbcomp#1{{\sf rbcomp}(#1)}
\def\vect{{\sf vect}}
  
%%%%%%%%%%%%%%%%%%%%%%%%%%%%%%%%%% 
%  Operations  
\let\cd=\cdot 
\let\eq=\equiv 
\let\op=\oplus 
\let\omn=\ominus
\let\meet=\wedge 
\let\join=\vee 
\let\tm=\times
\def\ldiv{\mathbin{\backslash}} 
\def\rdiv{\mathbin/}
  
%%%%%%%%%%%%%%%%%%%%%%%%%%%%%%%%%% 
%  Tame congruence  
\def\typ{{\sf typ}} 
\def\zz{{\un 0}} 
\def\zo{{\un 1}}
\def\one{{\bf1}} 
\def\two{{\bf2}} 
\def\three{{\bf3}}
\def\four{{\bf4}} 
\def\five{{\bf5}}
\def\pq#1{(\zz_{#1},\mu_{#1})}
  
%%%%%%%%%%%%%%%%%%%%%%%%%%%%%%%%%%% 
%  Accents and so  
\let\wh=\widehat 
\def\ox{\ov x} 
\def\oy{\ov y} 
\def\oz{\ov z}
\def\of{\ov f} 
\def\oa{\ov a} 
\def\ob{\ov b} 
\def\oc{\ov c}
\def\od{\ov d} 
\def\oov{\ov v} 
\def\ow{\ov w} 
\def\oob{\ov{\ov b}} 
\def\rx{{\rm x}}
\def\rf{{\rm f}} 
\def\rrm{{\rm m}} 
\let\un=\underline
\let\ov=\overline 
\let\cc=\circ 
\let\rb=\diamond 
\def\ta{{\tilde a}} 
\def\tz{{\tilde z}}
  
%%%%%%%%%%%%%%%%%%%%%%%%%%%%%%%%%%%%%%%%%% 
% Abbreviations for algebras and clones
  
\def\zZ{{\mathbb Z}} 
\def\B{{\mathcal B}} 
\def\P{{\mathcal P}}
\def\zL{{\mathbb L}} 
\def\zD{{\mathbb D}}
 \def\zE{{\mathbb E}}
\def\zG{{\mathbb G}} 
\def\zA{{\mathbb A}} 
\def\zB{{\mathbb B}}
\def\zC{{\mathbb C}} 
\def\zM{{\mathbb M}} 
\def\zR{{\mathbb R}}
\def\zS{{\mathbb S}} 
\def\zT{{\mathbb T}} 
\def\zN{{\mathbb N}}
\def\zQ{{\mathbb Q}} 
\def\zW{{\mathbb W}} 
\def\bK{{\bf K}}
\def\C{{\bf C}} 
\def\M{{\bf M}} 
\def\E{{\bf E}} 
\def\N{{\bf N}}
\def\O{{\bf O}} 
\def\bN{{\bf N}} 
\def\bX{{\bf X}} 
\def\GF{{\rm GF}} 
\def\cC{{\mathcal C}} 
\def\cA{{\mathcal A}}
\def\cB{{\mathcal B}} 
\def\cD{{\mathcal D}} 
\def\cE{{\mathcal E}} 
\def\cF{{\mathcal F}} 
\def\cG{{\mathcal G}} 
\def\cH{{\mathcal H}}
\def\cI{{\mathcal I}} 
\def\cL{{\mathcal L}} 
\def\cM{{\mathcal M}} 
\def\cP{{\mathcal P}} 
\def\cR{{\mathcal R}} 
\def\cRY{{\mathcal RY}}
\def\cS{{\mathcal S}} 
\def\cT{{\mathcal T}} 
\def\oB{{\ov B}}
\def\oC{{\ov C}} 
\def\ooB{{\ov{\ov B}}} 
\def\ozB{{\ov{\zB}}}
\def\ozD{{\ov{\zD}}} 
\def\ozG{{\ov{\zG}}}
\def\tcA{{\widetilde\cA}} 
\def\tcC{{\widetilde\cC}}
\def\tcF{{\widetilde\cF}} 
\def\tcI{{\widetilde\cI}}
\def\tB{{\widetilde B}} 
\def\tC{{\widetilde C}}
\def\tD{{\widetilde D}} 
\def\ttB{{\widetilde{\widetilde B}}}
\def\ttC{{\widetilde{\widetilde C}}}
\def\tba{{\tilde\ba}} 
\def\ttba{{\tilde{\tilde\ba}}}
\def\tbb{{\tilde\bb}} 
\def\ttbb{{\tilde{\tilde\bb}}}
\def\tbc{{\tilde\bc}} 
\def\tbd{{\tilde\bd}}
\def\tbe{{\tilde\be}} 
\def\tbt{{\tilde\bt}}
\def\tbu{{\tilde\bu}} 
\def\tbv{{\tilde\bv}}
\def\tbw{{\tilde\bw}} 
\def\tdl{{\tilde\dl}} 
\def\ocP{{\ov\cP}}
\def\tzA{{\widetilde\zA}} 
\def\tzC{{\widetilde\zC}}
\def\new{{\mbox{\footnotesize new}}}
\def\old{{\mbox{\footnotesize old}}}
\def\prev{{\mbox{\footnotesize prev}}}
\def\oo{{\mbox{\sf\footnotesize o}}}
\def\pp{{\mbox{\sf\footnotesize p}}}
\def\nn{{\mbox{\sf\footnotesize n}}} 
\def\oR{{\ov R}}
  
%%%%%%%%%%%%%%%%%%%%%%%%%%%%%%%%%%%%%%% 
% Abbreviations for varieties
  
\def\gA{{\mathfrak A}} 
\def\gV{{\mathfrak V}} 
\def\gS{{\mathfrak S}} 
\def\gK{{\mathfrak K}} 
\def\gH{{\mathfrak H}}
  
%%%%%%%%%%%%%%%%%%%%%%%%%%%%%%%%%%%%%%%% 
%  Vectors  
\def\ba{{\bf a}} 
\def\bb{{\bf b}} 
\def\bc{{\bf c}} 
\def\bd{{\bf d}} 
\def\be{{\bf e}} 
\def\bbf{{\bf f}} 
\def\bg{{\bf g}}
\def\bh{{\bf h}}
\def\bi{{\bf i}} 
\def\bo{{\bf o}} 
\def\bs{{\bf s}} 
\def\bu{{\bf u}} 
\def\bt{{\bf t}} 
\def\bv{{\bf v}} 
\def\bx{{\bf x}}
\def\by{{\bf y}} 
\def\bw{{\bf w}} 
\def\bz{{\bf z}}
\def\ga{{\mathfrak a}} 
\def\oal{{\ov\al}} 
\def\obeta{{\ov\beta}}
\def\ogm{{\ov\gm}} 
\def\oep{{\ov\varepsilon}}
\def\oeta{{\ov\eta}} 
\def\oth{{\ov\th}} 
\def\ovm{{\ov\mu}}
\def\ozero{{\ov0}}
  
%%%%%%%%%%%%%%%%%%%%%%%%%%%%%%%%%%%%%%%% 
% Constraint satisfaction Problem
  
\def\CCSP{\hbox{\rm c-CSP}} 
\def\CSP{{\rm CSP}} 
\def\NCSP{{\rm \#CSP}} 
\def\mCSP{{\rm MCSP}} 
\def\FP{{\rm FP}} 
\def\PTIME{{\bf PTIME}} 
\def\GS{\hbox{($*$)}} 
\def\ry{\hbox{\rm r+y}}
\def\rb{\hbox{\rm r+b}} 
\def\Gr#1{{\mathrm{Gr}(#1)}}
\def\Grp#1{{\mathrm{Gr'}(#1)}} 
\def\Grpr#1{{\mathrm{Gr''}(#1)}}
\def\Scc#1{{\mathrm{Scc}(#1)}} 
\def\rel{R} 
\def\relo{Q}
\def\rela{S} 
\def\dep{\mathsf{dep}}
\def\Filt#1{\mathsf{Ft}(#1)}
\def\Filts{\mathrm{Fts}} 
\def\Agr{$\mathbb{A}$}
\def\Al{\mathrm{Alg}}
\def\Alg{\mathrm{Alg}}
\def\Sig{\mathrm{Sig}}
\def\strat{\mathsf{strat}}
\def\relmax{\mathsf{relmax}}
\def\srelmax{\mathsf{srelmax}}
\def\Meet{\mathsf{Meet}}
\def\amax{\mathsf{amax}}
\def\as{\mathsf{as}}
\def\smax#1{\mathsf{max}(#1)}
\def\pmax#1{\mathsf{pmax}(#1)}
\def\smin{{\sf min}}
\def\cPmax{\cP^{\mathsf{max}}}
\def\star{\hbox{$\mathbf{(*)}$}}
\def\link{\mathsf{emin}}
\def\mal{{m}}
\def\bmal{\mathbf{m}}
\def\bot{\mathsf{bot}}
\def\down{\mathsf{down}}
\def\next{\mathsf{next}}

%%%%%%%%%%%%%%%%%%%%%%%%%%%%%%%%%%%%%%%% 
% Mathematical abbreviations
  
\let\sse=\subseteq 
\def\ang#1{\langle #1 \rangle}
\def\angg#1{\left\langle #1 \right\rangle}
\def\dang#1{\ang{\ang{#1}}} 
\def\vc#1#2{#1 _1\zd #1 _{#2}}
\def\zd{,\ldots,} 
\let\bks=\backslash 
\def\red#1{\vrule height7pt depth3pt width.4pt
\lower3pt\hbox{$\scriptstyle #1$}}
\def\fac#1{/\lower2pt\hbox{$\scriptstyle #1$}}
\def\me{\stackrel{\mu}{\eq}} 
\def\nme{\stackrel{\mu}{\not\eq}}
\def\eqc#1{\stackrel{#1}{\eq}} 
\def\cl#1#2{\arraycolsep0pt
\left(\begin{array}{c} #1\\ #2 \end{array}\right)}
\def\cll#1#2#3{\arraycolsep0pt \left(\begin{array}{c} #1\\ #2\\
#3 \end{array}\right)} 
\def\clll#1#2#3#4{\arraycolsep0pt
\left(\begin{array}{c} #1\\ #2\\ #3\\ #4 \end{array}\right)}
\def\cllll#1#2#3#4#5#6{ \left(\begin{array}{c} #1\\ #2\\ #3\\
#4\\ #5\\ #6 \end{array}\right)} 
\def\pr{{\rm pr}}
\def\perm#1#2#3{\left(\begin{array}{ccc} 1&2&3\\ #1&#2&#3
\end{array}\right)} 
\def\w{$\wedge$} 
\let\ex=\exists
\def\Id#1{{{\sf Id}(#1)}}
\def\tms{\tm\dots\tm}

%%%%%%%%%%%%%%%%%%%%%%%%%%%%%%%%%%%%%%%%%%% 
% Other abbreviations
  
\def\lb{$\linebreak$}  
  
%%%%%%%%%%%%%%%%%%%%%%%%%%%%%%%%%%%%%%%%%%% 
%  Functions  
\def\ar{\hbox{ar}} 
\def\Im{{\sf Im}} 
\def\deg{{\sf deg}}
\def\id{{\rm id}}
  
%%%%%%%%%%%%%%%%%%%%%%%%%%%%%%%%%%%%%%%%%% 
%  Greek symbols  
\let\al=\alpha 
\let\gm=\gamma 
\let\dl=\delta 
\let\ve=\varepsilon
\let\ld=\lambda 
\let\om=\omega 
\let\vf=\varphi 
\let\vr=\varrho
\let\th=\theta 
\let\sg=\sigma 
\let\Gm=\Gamma 
\let\Dl=\Delta
\let\kp=\kappa  

%%%%%%%%%%%%%%%%%%%%%%%%%%%%%%%%%%%%%%%%%%% 
% Fonts and special symbols
  
\font\tengoth=eufm10 scaled 1200 
\font\sixgoth=eufm6
\def\goth{\fam12} 
\textfont12=\tengoth 
\scriptfont12=\sixgoth
\scriptscriptfont12=\sixgoth 
\font\tenbur=msbm10
\font\eightbur=msbm8 
\def\bur{\fam13} 
\textfont11=\tenbur
\scriptfont11=\eightbur 
\scriptscriptfont11=\eightbur
\font\twelvebur=msbm10 scaled 1200 
\textfont13=\twelvebur
\scriptfont13=\tenbur 
\scriptscriptfont13=\eightbur
\mathchardef\nat="0B4E 
\mathchardef\eps="0D3F

%%%%%%%%%%%%%%%%%%%%%%%%%%%%%%%%%%%
%%%%%%%%%%%%%%%%%%%%%%%%%%%%%%%%%%% 
\title{Constraint Satisfaction Problems over semilattice block Mal'tsev algebras}
\author{Andrei A.\ Bulatov
%% School of Computing Science, 
%% Simon Fraser University, Burnaby, Canada\\ 
%% \it e-mail: abulatov@cs.sfu.ca} 
} 
\date{}
\maketitle

\begin{abstract}
There are two well-known types of algorithms for solving CSPs: local propagation and 
generating a basis of the solution space. For several years the focus of the CSP
research has been on `hybrid' algorithms that somehow combine the two approaches.
In this paper we present a new method of such hybridization that allows us
to solve certain CSPs that has been out of reach for a quite a while. 

We apply this method to CSPs parametrized by a universal algebra, an approach 
that has been very popular in the last decade or so. Specifically, we consider 
a fairly restricted class of algebras we will call semilattice block Mal'tsev. 
An algebra $\zA$ is called semilattice block Mal'tsev if it has a binary operation $f$, 
a ternary operation $m$, and a congruence $\sg$ such that the quotient $\zA\fac\sg$ 
with operation $f$ is a semilattice, $f$ is a projection on every block of $\sg$, 
and every block of $\sg$ is a Mal'tsev algebra with Mal'tsev operation $m$. This
means that the domain in such a CSP is partitioned into blocks such that
if the problem  is considered on the quotient set $\zA\fac\sg$, 
it can be solved by a simple
constraint propagation algorithm. On the other hand, if the problem is restricted on individual blocks, it can be solved by generating a basis of the 
solution space. We show that the two methods can be combined in a highly
nontrivial way, and therefore the constraint satisfaction problem over a 
semilattice block Mal'tsev algebra is solvable in polynomial time.
\end{abstract}

%%%%%%%%%%%%%%%%%%%%%%%%%%%%%%%%%%%
%%%%%%%%%%%%%%%%%%%%%%%%%%%%%%%%%%%
\section{Introduction}

% general stuff on complexity of CSP and statement
In a Constraint Satisfaction Problem (CSP, for short) we need to decide whether
or not a given set of constraints on values that can be assigned simultaneously
to a given set of variables can be satisfied. While the general CSP is 
NP-complete, its versions restricted by specifying a constraint language, a set
of allowed constraints, are sometimes solvable in polynomial time. For a 
constraint language $\Gm$ the corresponding restricted CSP is denoted
$\CSP(\Gm)$ and called a nonuniform CSP.
The study of the complexity of nonuniform CSPs
has been initiated by Schaefer \cite{Schaefer78:complexity}. In that paper 
Schaefer determined 
the complexity of $\CSP(\Gm)$ for constraint languages on a 2-element set. The 
complexity of $\CSP(\Gm)$ for constraint languages over 
finite sets has
been attracting much attention since then. This research is guided
by the Dichotomy Conjecture proposed by Feder and Vardi 
\cite{Feder93:monotone,Feder98:monotone} that states that every CSP of the
form $\CSP(\Gm)$ for a constraint language $\Gm$ on a finite set is either 
solvable in polynomial time or is NP-complete. The Dichotomy Conjecture 
has been restated and made more precise in different languages, see, e.g.\  
\cite{Bulatov05:classifying,Nesetril10:combinatorial}. Also, several powerful 
approaches to the problem have been developed, through algebra, 
logic, and graph theory. So far the most successful method of studying the 
complexity 
of the CSP has been the algebraic approach introduced by Jeavons et al.\ 
\cite{Bulatov03:multi,Bulatov05:classifying,Bulatov08:recent,Jeavons97:closure}. 
This approach relates the complexity of $\CSP(\Gm)$ to the properties of a 
certain universal algebra $\zA_\Gm$ associated with $\Gm$. In particular 
it allows one to expand $\CSP(\Gm)$ to the problem $\CSP(\zA_\Gm)$ that 
depends only on the associated algebra, without changing its complexity. It 
therefore suffices to restrict ourselves to 
the study of the complexity of problems of the form $\CSP(\zA)$, where $\zA$ is
a finite universal algebra.

% history, types of algorithms
The dichotomy conjecture  has been confirmed 
in a number of cases: for constraint languages on 2- and 3-element sets
\cite{Bulatov06:3-element,Schaefer78:complexity} (a dichotomy result was also 
announced for languages over 4-, 5-, and 7-element sets 
\cite{Markovic11:4-element,Zhuk16:5-element,Zhuk16:7-element}), 
for constraint languages containing all unary relations 
\cite{Barto11:conservative,Bulatov11:conservative,Bulatov16:conservative}, and
several others, see, e.g.\ \cite{Barto14:local,Barto09:sources,Idziak10:few}. 
One of the most remarkable phenomena discovered is that, generally, there are only
two types of algorithms applicable to CSPs solvable in polynomial time. The first
one has long been known to researchers in Artificial Intelligence as constraint 
propagation \cite{Dechter03:constraint}.  Algorithms of the other type resemble 
Gaussian elimination in the sense that they construct a small generating set of 
the set of all solutions \cite{Bulatov06:simple,Idziak10:few}. The scope of both 
types of algorithms is precisely known \cite{Barto14:local,Idziak10:few}.

% combinations 
General dichotomy results, however, cannot be proved using only algorithms of 
a single `pure' type. In all such results, see, e.g.\ \cite{Barto11:conservative,%
Bulatov06:3-element,Bulatov11:conservative,Bulatov16:conservative} a certain
mix of the two types of algorithms is needed. In some cases, for instance, 
\cite{Bulatov06:3-element} such a hybrid algorithm is somewhat ad hoc; 
in other cases, 
\cite{Barto11:conservative,Bulatov11:conservative,Bulatov16:conservative} 
it is based on intricate decompositions of the problem instance. It has become clear 
however that ad hoc hybridization and the decomposition techniques 
developed in the mentioned works are not sufficient. Therefore trying
to identify new polynomial time solvable cases of the CSP through combining
the two types of algorithms is the key to approaching the Dichotomy Conjecture.
There have been several further attempts to design hybrid algorithms; 
however, most of them were not quite successful. In more successful 
cases such as
\cite{Markovic09:block,Maroti11:Malcev,Maroti10:tree,Payne16:product} 
the researchers tried to tackle somewhat limited cases, in which a 
combination of local consistency properties and Gaussian elimination 
type fragments is very explicit. To provide the context for 
our results we explain those cases in details.

Suppose that a constraint language $\Gm$ is such that it is possible to 
partition its domain $A$ into blocks with the property that the restriction
of $\CSP(\Gm)$ on each block of the partition can be solved by an
algorithm of one type; while if we collapse each block into a single element,
the resulting quotient problem can be solved by an algorithm of another type.
What can be said about $\CSP(\Gm)$ itself? For instance, consider constraint
language $\Gm=\{\rel\}$ on $A=\{0,1,2\}$ where the ternary relation
$\rel$ is given by (triples in $\rel$ are written vertically)
$$
\rel=\left(\begin{array}{cccc|cccc|cc|cc}
0&0&1&1&2&2&2&2&2&2&2&2\\
0&1&0&1&0&0&1&1&2&2&0&1\\
0&1&1&0&1&0&0&1&0&1&2&2
\end{array}\right).
$$
If $A$ is partitioned into $B=\{0,1\}$ and $C=\{2\}$, then the restriction 
of $\rel$ on the blocks $B,C$ is one of the relations above separated 
by vertical lines (we can choose between $B$ and $C$ for different 
coordinate positions), and the corresponding CSP can be solved by 
Gaussian elimination. Indeed, 
the only nontrivial relation obtained this way is the first one, that is, 
$\rel\cap B^3$, and it is given by a linear equation $x+y+z=0$. 
The quotient relation $\rel'$ then looks like 
$$
\rel'=\left(\begin{array}{cccc}
B&C&C&C\\ B&B&C&B\\ B&B&B&C
\end{array}\right),
$$
and it follows from \cite{Schaefer78:complexity} that $\CSP(\rel')$ can be
solved by a local propagation algorithm, as $\rel'$ can be represented by a
Horn clause. Solving $\CSP(\Gm)$ is less easy, see, \cite{Bulatov06:3-element},
and similar but more complicated cases have not been known to be
polynomial time solvable until now.

To make constructions like the one above more precise we use the algebraic 
representation of nonuniform CSPs, in which a 
constraint language is replaced with its (universal) algebra of polymorphisms. 
This allows us to exploit structural properties of algebras to design 
a hybrid algorithm. So, starting from $\CSP(\Gm)$, where $\Gm$ is a 
constraint language on a set $A$, we first consider the corresponding 
algebra $\zA_\Gm$ with base set $A$ such $\CSP(\zA_\Gm)$ is polynomial 
time reducible to $\CSP(\Gm)$. A partition of $\zA_\Gm$ is given by a 
congruence of $\zA_\Gm$, that is, an invariant equivalence relation. Recall that 
due to the results of \cite{Bulatov05:classifying} the algebra $\zA_\Gm$ can be 
assumed idempotent, this makes restrictions on congruence blocks possible.
Now, suppose that an idempotent algebra $\zA$ is such that it has a congruence 
$\sg$ with the property that the CSP of its quotient $\zA\fac\sg$ can be 
solved by the small generating set algorithm, say, it is Mal'tsev, while for 
every $\sg$-block $\zB$ (a subalgebra of $\zA$) the CSP over $\zB$ can be 
solved by a local propagation algorithm; or the other way round, see 
Figure~\ref{fig:blocks}. 
How can one solve the CSP over $\zA$ itself? Maroti in \cite{Maroti11:Malcev} 
considered the first case, when $\zA\fac\sg$ can be solved 
by the small generating set algorithm. This case turns 
out to be easier because of the property of the $\sg$-blocks we can exploit. 
Suppose for simplicity that every $\sg$-block $\zB$ is a semilattice, as 
shown in Figure~\ref{fig:blocks}. Then 
every CSP instance on $\zB$ has some sort of a canonical solution that assigns
the maximal element of the semilattice (that is element $a\in\zB$ such that 
$ab=a$ for all $b\in\zB$) to every variable. It then can be shown that if we find a solution 
$\vf:V\to\zA\fac\sg$ where $V$ is the set of variables of the instance on  
$\zA\fac\sg$, and then assign the maximal elements of the 
$\sg$-block $\vf(v)$ to $v$, we obtain a solution of the original instance. 

\begin{figure}[h]
%%\vspace*{-10mm}
\centerline{\includegraphics[scale=0.6]{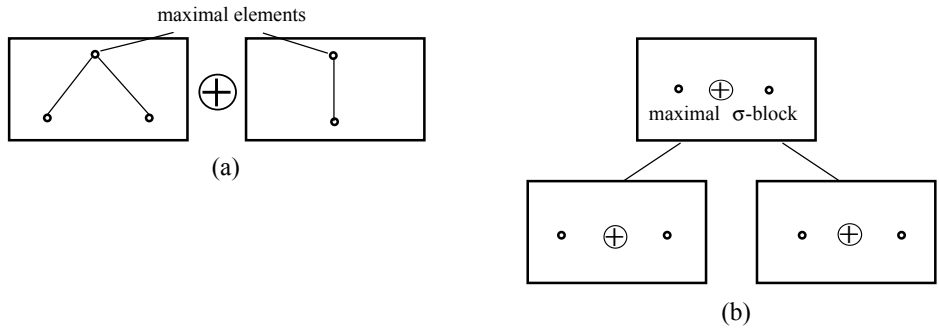}}
\caption{(a) Algebra $\zA$ such that $\zA\fac\sg$ is Mal'tsev; (b) an SBM
algebra. 
Rectangles represent $\sg$-blocks, dots represent elements,
lines show the semilattice structure, and $\oplus$ represents a Mal'tsev 
operation acting on elements or $\sg$-blocks.}\label{fig:blocks}
\end{figure}

The case when $\zA\fac\sg$ is a semilattice, while every $\sg$-block is Mal'tsev 
is much more difficult. We will call such algebras \emph{semilattice block 
Mal'tsev} algebras (SBM algebras, for short). More precisely,
we consider idempotent algebras $\zA$ with the following property: 
There are a binary operation $f$ and a ternary operation $m$, and a congruence 
$\sg$ of $\zA$ such that $\zA\fac\sg$ is a semilattice with a semilattice 
operation $f$, and every $\sg$-block $B$ is a Mal'tsev algebra with Mal'tsev 
operation $m$, and $f\red B$ is a projection. The main difficulty with this 
kind of algebras is that the only solution of a CSP over a semilattice we can 
reliably find is the canonical one assigning the maximal available element 
to each variable. Finding 
a second solution is already hard. On the other hand, 
if we restrict our instance only to the maximal $\sg$-block $\zB$, it may have no 
solution there, even though the original instance has a solution, which 
simply does not belong to the maximal block. If this is the case, it has been 
unclear for nearly 10 years how  the domain can be reduced so that the maximal 
block is eliminated.

The problem has been resolved in some special cases. Firstly, Maroti in 
\cite{Maroti10:tree} showed that it suffices to consider SBM algebras of a certain
restricted type. We will use this result in this paper. Marcovi\'c and McKenzie 
suggested an algorithm that solves 
the CSP over SBM algebras $\zA$ when $\zA\fac\sg$ is a chain, that is, 
$ab\in\{a,b\}$ 
for any $a,b\in\zA\fac\sg$. In this case their algorithm is capable of eliminating 
the maximal block using the fact that if a semilattice is a chain, any 
of its subsets is a 
subalgebra. Finally, very recently Payne in \cite{Payne16:product} suggested an 
algorithm that works for a more general class of algebras than SBM, but algebras 
in this class have to satisfy an extra condition that in SBM algebras manifests 
itself as the existence of certain well behaving mappings between 
$\sg$-blocks. In particular, this condition guarantees that the instance 
restricted to the maximal $\sg$-block has a solution whenever the original 
problem has a solution.

% semilattice block Mal'tsev
In this paper we continue the effort started in 
\cite{Markovic09:block,Maroti10:tree,Payne16:product} and
present an algorithm that solves the CSP over an arbitrary SBM algebra.

\begin{theorem}\label{the:main}
If $\zA$ is a SBM algebra then $\CSP(\zA)$ is solvable in polynomial time.
\end{theorem}

The algorithm is based upon a new local consistency notion that we call
\emph{block-minimality} (although in our case it is necessarily not quite 
local, since it has to deal with Mal'tsev algebras). More specifically, our 
algorithm first separates the set $V$ of variables of a CSP instance into 
overlapping subsets, coherent sets, and considers subproblems on these 
sets of variables. 
For block-minimality these subproblems have to be minimal, that is, every 
tuple from every constraint relation has to be a part
of a solution. This can be achieved by solving the problem many times with 
additional constraints. However, this is not very straightforward, because 
coherent sets may contain all the variables from $V$. To overcome this 
problem we show that the subproblems
restricted to coherent sets are either over a Mal'tsev domain and therefore can 
be solved efficiently, or they split up into a collection of disjoint instances,
each of which has a strictly smaller domain. In the latter case we can recurse
on these smaller instances.
Finally, we prove that any block-minimal instance has a solution.

The results of this paper can easily be made more general by removing some of 
the restrictions on the basic operations of SBM algebras. However, we hope that 
these results can be generalized well beyond SBM-like algebras and so 
we stop short of giving more general but also more technically involved 
proofs just restricting ourselves to demonstrating the general idea.

% overview
In Section~\ref{sec:prelims} we recall the basic definitions on CSP and the 
algebraic approach. A somewhat simplified outline of the solution 
algorithm and block-minimality is given in Section~\ref{sec:outline}.
More advanced facts from algebra and a study of certain 
properties of SBM algebras are given in Section~\ref{sec:SBM-minimal}.
In Section~\ref{sec:coherent} we
strengthen the results of \cite{Bulatov02:maltsev-3-element} about the 
structure of relations over Mal'tsev algebras and extend them to SBM 
algebras\footnote{Kearnes and Szendrei in \cite{Kearnes12:clones} developed
a technique based on so-called critical relations that resembles in certain 
aspects what can be achieved through coherent sets. However, 
\cite{Kearnes12:clones} only concerns congruence modular algebras, and so
cannot be used for SBM algebras.}. In Section~\ref{sec:to-instances} we 
extend these notions to CSP instances. Finally, in Section~\ref{sec:algorithm} 
we prove the main result and present a solution algorithm.

%%%%%%%%%%%%%%%%%%%%%%%%%%%%%%%%%%%
%%%%%%%%%%%%%%%%%%%%%%%%%%%%%%%%%%%
\section{Preliminaries}\label{sec:prelims}

%%%%%%%%%%%%%%%%%%%%%%%%%%%%%%%%%%%
\subsection{Multisorted Constraint Satisfaction Problem}

%basic notation
By $[n]$ we denote the set $\{1\zd n\}$. Let $\vc An$ be finite sets. Tuples from
$A_1\tm\ldots\tm A_n$ are denoted in boldface, say, $\ba$, and their entries by
$\ba[1]\zd\ba[n]$.  A \emph{relation} $\rel$ over $\vc An$ is a subset of 
$A_1\tm\dots\tm A_n$. We refer to $n$ as the \emph{arity} of the tuple $\ba$ 
and the relation $\rel$. Let $I=(\vc ik)$ be 
an (ordered) multiset, a subset of $[n]$. Then let $\pr_I\ba=(\ba[i_1]\zd\ba[i_k])$ 
and $\pr_I\rel=\{\pr_I\ba\mid\ba\in\rel\}$. Relation $\rel$ is said to be a 
\emph{subdirect product} of $\vc An$ if $\pr_i\rel=A_i$ for $i\in [n]$. In some cases
it will be convenient to consider tuples and relations whose entries are indexed by 
sets other than subsets of $[n]$, most often those will be sets of variables. Then we either 
assume the index set is somehow ordered, or consider tuples as functions from
the index set to the domain and relations as sets of such functions.

%multisorted CSP
Let $\cA$ be a set of sets, in this paper $\cA$ is usually the set of universes of
finite algebras derived from an SBM algebra; we clarify `derived' later.
An instance of a \emph{(Multisorted) Constraint Satisfaction Problem} (CSP) 
over $\cA$ is given by $\cP=(V,\cA,\cC)$, where $V$ is a set of 
\emph{variables}, $\cA$ is a collection of \emph{domains} $A_v\in\cA$, and 
$\cC$ is a set of \emph{constraints}; every 
constraint $\ang{\bs,\rel}$ is a pair consisting of an ordered multiset 
$\bs=(\vc vk)$, a subset of $V$, called the \emph{constraint scope} and 
$\rel$, a  relation over $A_{v_1}\zd A_{v_k}$, called the 
\emph{constraint relation}.

%%%%%%%%%%%%%%%%%%%%%%%%%%%%%%%%%%
\subsection{Algebraic structure of the CSP}

For a detailed introduction to CSP and the algebraic approach to its structure
the reader is referred to a very recent and very nice survey by Barto et al.\
\cite{Barto17:polymorphisms}. Basics of universal algebra can be learned from
the textbook \cite{Burris81:universal} and monograph \cite{Hobby88:structure}.

A (\emph{universal}) \emph{algebra} is a pair $\zA=(A;F)$, where $A$ is a set
(always finite in this paper) called the \emph{universe} of $\zA$, and $F$ is a 
set of \emph{basic operations}, multi-ary operations on $A$. Algebras 
$\zA=(A,F^\zA)$ and $\zB=(B,F^\zB)$ are said to be \emph{similar} 
if their basic operations are indexed by elements of the same set $F$ in such a 
way that operations from $F^\zA$ and $F^\zB$ indexed by the same element 
have the same arity. Operations that can be obtained from the basic operations 
of $\zA$ or a class $\gA$ of similar algebras by means of
compositions are said to be \emph{term} operations of $\zA$ or, respectively,
$\gA$.

The CSP is related to algebras through the notion of polymorphism. Let 
$\rel$ be a relation on a set $A$ and $f$ is a $k$-ary operation on the same 
set. Operation $f$ is said to be a \emph{polymorphism} of $\rel$ if for any
$\vc\ba k\in\rel$ the tuple $f(\vc\ba k)$ also belongs to $\rel$. More 
generally, let $\rel$ be a subset of $A_1\tms A_\ell$ and $f$ is an operation
symbol such that $f^{\zA_i}$ is a $k$-ary operation on $A_i$ for $i\in[\ell]$. 
Then $f$ is a polymorphism of $\rel$ if for any $\vc\ba k\in\rel$ the tuple
$f(\vc\ba k)$ belongs to $\rel$, where 
$f(\vc\ba k)=
(f^{\zA_1}(\ba_1[1]\zd\ba_k[1])\zd f^{\zA_\ell}(\ba_1[\ell]\zd\ba_k[\ell]))$.
Let $\Gm$ be a constraint language on a set $A$. Then $\Pol(\Gm)$ denotes 
the set of all operations $f$ on $A$ such that $f$ is a polymorphism of 
every relation from $\Gm$; also $\zA_\Gm=(A,\Pol(\Gm))$ is the corresponding
algebra. Similarly, let $\cA$ be a collection of sets and $\Gm$ a constraint
language over $\cA$, that is, a set of relations $\rel\sse A_1\tms A_\ell$,
$\vc A\ell\in\cA$. Then $F=\Pol(\Gm)$ is the set of all operation symbols $f$
along with their interpretations on sets from $\cA$ such that $f$ is a 
polymorphism of all relations from $\Gm$. The corresponding set of 
algebras is denoted by $\gA_\Gm$, that is, for every $A\in\cA$ the set
$\gA_\Gm$ contains algebra $\zA=(A,F^\zA)$, where 
$F^\zA=\{f^\zA\mid f\in F\}$. 

%CSP and algebras
Any class of similar algebras also gives rise to a CSP. Let $\gA$ be a 
class of similar finite algebras and $\cA$ the set of 
universes of algebras from $\gA$. Then $\CSP(\gA)$ is the class of instances 
$(V,\cA,\cC)$ of CSPs over $\cA$ such that every constraint relation $\rel$ 
from $\ang{\bs,\rel}\in\cC$, $\bs=(\vc vk)$, 
is a subalgebra of $A_{v_1}\tm\dots\tm A_{v_k}$, where $A_v$, $v\in V$, 
are viewed as algebras from $\gA$. 

In this paper we will use two special types of operations.

\begin{example}\label{exa:semilattice}\rm
A binary operation $f$ on $A$ is said to be \emph{semilattice} if 
$f(a,a)=a$, $f(a,b)=f(b,a)$, and $f(f(a,b),c)=f(a,f(b,c))$ for any 
$a,b,c\in A$. Similarly, $f$ is a semilattice operation on a class $\gA$
of similar algebras, if it is a term operation of that class and 
$f^\zA$ is a semilattice operation 
for every $\zA\in\gA$. We will treat a semilattice operation as multiplication
and denote it by $\cdot$ or omit the sign altogether. A semilattice operation 
defines an order on its domain: $a\le b$ if and only if $ab=b$. This means
that there is always the greatest element of such a semilattice order ---
the product of all the elements of $A$. We will denote this element by 
$\max(\zA)$.
\end{example}

\begin{example}\label{exa:maltsev}\rm
A ternary operation $\mal$ is said to be \emph{Mal'tsev} if it satisfies
the equations $\mal(a,b,b)=\mal(b,b,a)=a$ for any $a,b\in A$. A term 
operation $\mal$ of a class $\gA$ is Mal'tsev if $\mal^\zA$ is Mal'tsev for
every $\zA\in\gA$. An algebra with a Mal'tsev term operation is said to be
\emph{Mal'tsev}.

If $\gA$ has a Mal'tsev term operation, the algorithm from 
\cite{Bulatov06:simple} constructs a compact representation of the set
of solutions of any instance from $\CSP(\gA)$, thus solving the problem in 
polynomial time.
\end{example}

A \emph{subalgebra} of an algebra $\zA=(A,F)$ is a subset $B\sse A$ 
equipped with the restrictions of operations from $F$ on $B$ and such that 
$f(\vc ak)\in B$ for every $f\in F$ and $\vc ak\in B$.
An equivalence relation on $A$ invariant with respect to the basic operations
of $\zA$ is said to be a \emph{congruence} of $\zA$. If $a,b$ are related by 
a congruence $\al$, we write $a\eqc\al b$; the $\al$-block containing $a$ 
is denoted $a^\al$.  The \emph{quotient algebra} $\zA\fac\al$ has the 
universe $A\fac\al$ and basic operations $f^\al$, $f\in F$, such that 
for any $\vc ak\in A$ operation $f^\al$ is given by
$f^\al(a_1^\al\zd a_k^\al)=(f(\vc ak))^\al$. We will omit the superscript in 
$f^\al$ whenever this does not lead to a confusion.
Algebra $\zA$ is said to be \emph{idempotent} if $f(a\zd a)=a$ for any
$f\in F$ and any $a\in A$. A useful property of idempotent algebras is
that every class of any of its congruences is a subalgebra. In particular,
every 1-element subset of $A$ is a subalgebra.
Algebras $\zA,\zA'$ with the same universe are called \emph{term equivalent}
if they have the same set of term operations. If $\zA=(A,F)$, $\zA'=(A,F')$
and $F'$ is a subset of the set of term operations of $\zA$, then $\zA'$
is said to be a \emph{reduct} of $\zA$.

Idempotent algebra $\zA$ is said to be \emph{semilattice block Mal'tsev} 
if there are a binary term operation $f$ and a ternary term operation $\mal$, 
and a congruence $\sg$ of $\zA$ such that $\zA\fac\sg$ is term equivalent 
to a semilattice with a semilattice operation $f$, operation $\mal$ is a Mal'tsev 
operation on every $\sg$-block $B$, and $f\red B$ is a projection, that is,
$f\red B(x,y)=x$. 

%%%%%%%%%%%%%%%%%%%%%%%%%%%%%%%%%%
\subsection{Partial solutions and local consistency}

Let $\cP=(V,\cA,\cC)$ be a CSP instance
% partial solutions
Let $W\sse V$. By $\cP_W$ we denote the instance $(W,\cA^W,\cC^W)$ 
defined as follows: $A^W_v=A_v$ for each $v\in W$; for every constraint 
$C=\ang{\bs,\rel}$, $C\in\cC$, the set $\cC^W$ includes the constraint 
$C^W=\ang{\bs',\rel'}$, where $\bs'=\bs\cap W$ and $\rel'=\pr_{\bs'}\rel$. 
A solution of $\cP_W$ is called a \emph{partial solution} of $\cP$ on $W$. 
The set of all such solutions is denoted by $\cS_W$. If $W=\{v\}$ or 
$W=\{u,v\}$, we simplify notation to $\cP_v,\cS_v$ and 
$\cP_{uv},\cS_{uv}$, respectively.

% 3-minimal and minimal
Instance $\cP$ is called \emph{minimal} if every tuple $\ba\in\rel$ for any 
constraint $\ang{\bs,\rel}\in\cC$ can be extended to a solution of $\cP$; 
that is, there is $\vf\in\cS$ such that $\vf(v)=\ba[v]$ for $v\in\bs$. Instance 
$\cP$ is called \emph{$k$-minimal} if $\cP_W$ is minimal for all $k$-element 
$W\sse V$. For any fixed $k$ every instance can be reduced to a $k$-minimal 
instance in polynomial time by a standard algorithm \cite{Bulatov08:dualities}: 
cycle over all $k$ element subsets
$W\sse V$, solve the problem $\cP_W$, and for every constraint $\ang{\bs,\rel}$
exclude from $\rel$ all tuples inconsistent with $\cS_W$. If $\cP\in\CSP(\gA)$ 
for some class $\gA$ of similar algebras closed under subalgebras, 
the resulting problem also belongs to $\CSP(\gA)$.
In particular, from now on we will assume that all the instances we deal with are 
1-minimal. For such problems we can also \emph{tighten} the instance reducing 
the domains $A_v$, $v\in  V$, to the sets $\cS_v$. Every constraint relation will 
therefore be assumed to be a subdirect product of the respective domains.
If $\gA$ consists of idempotent algebras, then any problem from $\CSP(\gA)$ can be 
reduced to a minimal one by solving polynomially many instances of $\CSP(\gA)$. 
First of all, \emph{constant relations}, $\rel_a=\{(a)\}$, $a\in\zA\in\gA$, are subalgebras
of $\zA$ and therefore can be used in constraints. Then the algorithm proceeds as
follows: cycle over all constraints $C=\ang{\bs,\rel}\in\cC$ and all $\ba\in\rel$; replace 
$C$ with the collection of unary constraints $\ang{(\bs[i]),\rel_{\ba[\bs[i]]}}$;
solve the resulting instance $\cP_{C,\ba}$; remove $\ba$ from 
$\rel$ if $\cP_{C,\ba}$ has no solutions. However, this procedure obviously
amounts to solving instances from $\CSP(\gA)$, and therefore there is no 
guarantee this can be done in polynomial time.

\begin{example}\label{exa:semilattice-minimal}\rm
If a class $\gA$ of similar algebras has a semilattice term operation 
then $\CSP(\gA)$ can
be solved by establishing 1-minimality. More precisely, if
$\cP=(V,\cA,\cC)$ is a 1-minimal instance from $\CSP(\gA)$, where
$A_v$ is the domain of $v\in V$, then the mapping $\vf(v)=\max(A_v)$
is a solution of $\cP$.
\end{example}

%%%%%%%%%%%%%%%%%%%%%%%%%%%%%%%%%
\subsection{Congruences and polynomials}\label{sec:congruences}

The set (lattice) of congruences of an algebra $\zA$ will be denoted by 
$\Con(\zA)$. So, $\Con(\zA)$ is equipped with two binary operations 
of \emph{join}, $\join$, and \emph{meet}, $\meet$. The smallest 
congruence of $\zA$, the equality relation, is denoted by $\zz_\zA$, and
the greatest congruence, the total relation, is denoted by $\zo_\zA$.   
Let $\rel$ be
a subdirect product of $\vc\zA k$, and $\al_i\in\Con(\zA_i)$, $i\in [k]$. Then by
$\ov\al_\rel$, or simply $\ov\al$ if $\rel$ is clear from the context, we denote the
congruence $\al_1\tm\dots\tm\al_k$ of $\rel$ given by $\ba\eqc{\ov\al}\bb$ if
and only if $\ba[i]\eqc{\al_i}\bb[i]$ for all $i\in[k]$. Also, if 
$I=\{\vc i\ell\}\sse[k]$ then by $\ov\al_I$ we denote the congruence 
$\al_{i_1}\tm\dots\tm\al_{i_\ell}$ of $\pr_I\rel$.

Let $\cP=(V,\cA,\cC)$ be an instance of $\CSP(\gA)$ and $\al_v$ a congruence 
of $\zA_v\in\gA$ for each $v\in V$. 
By $\cP_{\ov\al}$ we denote the instance $(V,\cA^{\ov\al},\cC^{\ov\al})$, 
in which $\zA_v^{\ov\al}=\zA_v\fac{\al_v}$,
and a constraint $\ang{\bs,\rel'}$, $\bs=(\vc vk)$, belongs to $\cC^{\ov\al}$ 
if and only if a constraint $\ang{\bs,\rel}$, where 
$$
\rel'=\rel\fac{\ov\al}=\{\ba^{\ov\al}=
(\ba[1]^{\al_{v_1}}\zd\ba[k]^{\al_{v_k}})\mid \ba\in\rel\},
$$
belongs to $\cC$. 

A pair of congruences $\al,\beta\in\Con(\zA)$ is said to be a \emph{prime 
interval}, denoted $\al\prec\beta$, if $\al\le\beta$ and $\al<\gm<\beta$ 
for no congruence $\gm\in\Con(\zA)$. Then $\al\preceq\beta$ means that 
$\al\prec\beta$ or $\al=\beta$. For an operation $f$ on $\zA$ we write 
$f(\beta)\sse\al$ if, for any $a,b\in\zA$ with $a\eqc\beta b$, $f(a)\eqc\al f(b)$. 

Polynomials of $\zA$ are
formed from term operations as follows. Let\lb $f(\vc xk,\vc y\ell)$ be a term
operation of $\zA$ and $\vc a\ell\in\zA$. Then the operation 
$g(\vc xk)=f(\vc xk,\vc a\ell)$ is said to be a \emph{polynomial} of $\zA$.
Note that although a polynomial does not have to be a polymorphism
of invariant relations of $\zA$, unary polynomials and congruences of
$\zA$ are in a special relationship: an equivalence relation $\al$ is
a congruence of $\zA$ if and only if it is preserved by every unary 
polynomial $f$, that is, $f(\al)\sse\al$. 
As usual, by an \emph{idempotent
unary polynomial} we mean a polynomial $f(x)$ such that $f\circ f=f$ or,
equivalently, such that $f(x)=x$ for any $x$ from its range.

Let $\rel$ be a subdirect product
of $\vc\zA k$. Similar to tuples from $\rel$, polynomials of $\rel$ are also denoted in 
boldface, say,~$\bbf$. A polynomial $\bbf$ can be represented as $\bbf(\vc xk)=g(\vc
xk,\ba^1\zd\ba^\ell)$ where $g$ is a term operation of $\rel$ and
$\ba^1\zd\ba^l\in\rel$. Then the polynomial $g(\vc xk,\ba^1[i]\zd \ba^\ell[i])$
of $\zA_i$ is denoted by $f_i$, and for $I=\{\vc is\}\sse[n]$, $\bbf_I$ denotes
the polynomial $g(\vc xk,\pr_I\ba^1\zd\pr_I\ba^\ell)$ of $\pr_I\rel$. For any $i$, 
and any polynomial $f$ of $\zA_i$, there is a polynomial $\bg$ of $\rel$ such that 
$g_i=f$. We shall call $\bg$ an {\em extension} of $f$ to a polynomial of
$\rel$. Finally, for $I\sse[k]$, and $\ba\in\prod_{i\in I}\zA_i$ and 
$\bb\in\prod_{i\in[k]-I}\zA_i$, $(\ba,\bb)$ denotes the tuple $\bc$ such that
$\bc[i]=\ba[i]$ for $i\in I$ and $\bc[i]=\bb[i]$ if $i\in[k]-I$. To distinguish such
concatenation of tuples from pairs of tuples, we will denote pairs of tuples by 
$\ang{\ba,\bb}$. 

The proposition below lists the main basic properties of relations over Mal'tsev 
algebras.

\begin{prop}[Folklore]\label{pro:basics}
Let $\rel$ be a subdirect product of Mal'tsev algebras
$\zA_1\tms\zA_k$ and $I\sse[k]$. Then the following properties hold\\
(1) $\rel$ is {\em rectangular}, that is if
$\ba,\bb\in\pr_I\rel, \bc,\bd\in\pr_{[k]-I}\rel$ and
$(\ba,\bc),(\ba,\bd)$, $(\bb,\bc)\in\rel$, then $(\bb,\bd)\in\rel$.\\
(2) The relation $\nu_I=\{\ang{\ba,\bb}\in(\pr_I\rel)^2\mid$ there
is $\bc\in\pr_{[k]-I}\rel$ such that $(\ba,\bc)$, $(\bb,\bc)\in\rel\}$ is a congruence 
of $\pr_I\rel$.
\end{prop} 

%%%%%%%%%%%%%%%%%%%%%%%%%%%%%%%%%%%
%%%%%%%%%%%%%%%%%%%%%%%%%%%%%%%%%%%
\section{Outline of the algorithm}\label{sec:outline}

Our solution algorithm works by establishing some sort of minimality condition
and repeatedly alternates two phases. The first phase is based on the results of
Maroti \cite{Maroti10:tree} that allow us to reduce an instance over
SBM algebras to one over SBM algebras with a \emph{minimal} element.
If $\zA$ is an SBM algebra then there is a congruence $\sg$ such that 
$\zA\fac\sg$ is a semilattice. This means that $\zA\fac\sg$ has a maximal or 
\emph{absorbing} element $a$ such that $ax=xa=a$ for any $x\in\zA\fac\sg$. 
This
element will be in the focus of our argument. We will also show with help of 
\cite{Maroti10:tree}, Corollary~\ref{cor:maroti}, that it can always be 
assumed that $\zA\fac\sg$ 
has a minimal or \emph{neutral} element $b$ such that $bx=xb=x$ for
any $x\in\zA\fac\sg$. In fact, one can assume an even stronger condition:
that $b$ is a 1-element $\sg$-block.

For the second phase we introduce the \emph{block-minimality} condition
defined with the help of congruences and polynomials of an algebra.
Let $\rel$ be a subdirect product of $\zA_1\tm\dots\tm\zA_n$ and 
$\al,\beta\in\Con(\zA_i)$, $\gm,\dl\in\Con(\zA_j)$ such that $\al\prec\beta$,
$\gm\prec\dl$ for some $i,j\in[n]$. Interval $(\al,\beta)$ 
\emph{can be separated} 
from $(\gm,\dl)$ if there is a unary polynomial $\bbf$ of $\rel$ such 
that $f_i(\beta)\not\sse\al$ while $f_j(\dl)\sse\gm$.
We are mostly
interested in the situation when prime intervals cannot be separated.

Suppose that $\cP=(V,\cA,\cC)$ is a 3-minimal instance and the domain
$\zA_v$ of $v\in V$ is an SBM algebra and $\sg_v$ is such that 
$\zA_v\fac{\sg_v}$ is a semilattice. Let $\th_v$ denote the congruence 
of $\zA_v$ such that the maximal element of $\zA_v\fac{\sg_v}$ is one 
block of $\th_v$, and all other $\th_v$-blocks are singletons. We show,
Lemma~\ref{lem:irreducible-maximal}, that this is indeed a congruence. 
For every $v\in V$ and $\al,\beta\in\Con(\zA_v)$ with $\al\prec\beta\le\th_v$
let $W_{v\al\beta}\sse V$ denote the set 
of variables $w$ such that $(\al,\beta)$ and $(\gm,\dl)$ for some 
$\gm,\dl\in\Con(\zA_w)$ with $\gm\prec\dl\le\th_w$ cannot be separated from 
each other in the binary relation $\cS_{vw}$. We call such sets of variables
\emph{coherent sets}. Instance $\cP$ is said to be
\emph{block-minimal} if for every $v\in V$ and $\al,\beta\in\Con(\zA_v)$ with 
$\al\prec\beta\le\th_v$ the problem $\cP_{W_{v\al\beta}}$ is minimal.

The result now follows from the following two statements. First, 
Proposition~\ref{pro:to-block-minimal} claims that any instance $\cP$
over SBM algebras can be efficiently reduced to an equivalent 
block-minimal instance by solving polynomially many SBM instances
over domains of smaller size. The second statement, 
Theorem~\ref{the:block-minimal}, claims that any block-minimal SBM 
instance has a solution. 

The key to the proof of Proposition~\ref{pro:to-block-minimal} is 
Lemma~\ref{lem:coherent-link-partition} stating that
every problem $\cP_{W_{v\al\beta}}$ is a disjoint union of problems over smaller 
domains, or its domains are Mal'tsev algebras. More precisely, in the first case 
there is $k$ such that for every $w\in W_{v\al\beta}$
the domain $\zA_w$ can be partitioned into a disjoint union 
$\zA_w^{(1)}\cup\dots\cup\zA_w^{(k)}$ in such a way that for any
constraint $\ang{(\vc v\ell),\rel}$ of $\cP_{W_{v\al\beta}}$, every tuple
$\ba\in\rel$ belongs to $\zA_{v_1}^{(j)}\tms\zA_{v_k}^{(j)}$ 
for some $j\in[k]$. This property follows from the existence of a minimal 
element in every domain and the fact that certain prime intervals in 
congruence lattices of the domains of $\cP_{W_{v\al\beta}}$ cannot be
separated from each other, Lemma~\ref{lem:coherent-link-partition}. 
It means, of course, that it suffices to solve 
$k$ problems $\cP_{W_{v\al\beta}}^{(j)}$ whose domains are $\zA_w^{(j)}$.

We prove Theorem~\ref{the:block-minimal} by induction, showing that 
for every $\ov\beta=(\beta_v)_{v\in V}$ with 
$\beta_v\in\Con(\zA_v)$ with $\beta_v\le\th_v$ there is a collection of 
solutions $\vf_{v\al\beta}$ of $\cP_{W_{v\al\beta}}$ such that whenever 
$u\in W_{v\al\beta}\cap W_{w\gm\dl}$
we have $\vf_{v\al\beta}(u)\eqc{\beta_u}\vf_{w\gm\dl}(u)$. If every $\beta_w$ 
equals $\th_w$ then such a collection exists because 
the maximal element of $\zA_w\fac{\beta_w}$ is a singleton, and we
always can choose mappings $\vf_{v\al\beta}$ to be such that 
$\vf_{v\al\beta}(w)\fac{\th_v}$ is the maximal element. On the other
hand, if $\beta_w$ is the equality relation for every $w\in V$ then 
solutions $\vf_{v\al\beta}$ agree with each other and provide a solution of $\cP$. 
Thus, showing that the existence of solutions $\vf_{v\al\beta}$ for some 
$\ov\beta$ implies the existence of such solutions for smaller congruences 
$\ov\beta'$ is the crux of our argument.

%%%%%%%%%%%%%%%%%%%%%%%%%%%%%%%%%%%
%%%%%%%%%%%%%%%%%%%%%%%%%%%%%%%%%%%
\section{Semilattice block Mal'tsev algebras and minimal elements}%
\label{sec:SBM-minimal}

%%%%%%%%%%%%%%%%%%%%%%%%%%%%%%%%%%%
\subsection{Minimal sets and polynomials}\label{sec:minimal-set}

We will use several basic concepts of the tame congruence theory, 
\cite{Hobby88:structure}.

An \emph{$(\al,\beta)$-minimal set} is a minimal (under inclusion) set $U$ 
such that $U=f(\zA)$ for a unary polynomial of $\zA$ satisfying 
$f(\beta)\not\sse\al$. Sets $B,C$ are said
to be \emph{polynomially isomorphic}\index{polynomially isomorphic sets}
in $\zA$ if there are unary polynomials $f,g$ such that $f(B)=C$, $g(C)=B$, and
$f\circ g, g\circ f$ are identity mappings on $C$ and $B$, respectively.

\begin{lemma}[Theorem 2.8, \cite{Hobby88:structure}]%
\label{lem:minimal-sets}
Let $\al,\beta\in\Con(\zA)$, $\al\prec\beta$. Then the following hold.\\[1mm]
(1) Any $(\al,\beta)$-minimal sets $U,V$ are polynomially isomorphic.\\[1mm]
(2) For any $(\al,\beta)$-minimal set $U$ and any unary polynomial 
$f$, if
$f(\beta\red U)\not\sse\al$ then $f(U)$ is an $(\al,\beta)$-minimal set, $U$ 
and $f(U)$ are polynomially isomorphic,  and $f$ witnesses this fact.\\[1mm]
(3) For any $(\al,\beta)$-minimal set $U$ there is a unary polynomial 
$f$ such that
$f(\zA)=U$, $f(\beta)\not\sse\al$, and $f$ is idempotent, in particular, 
$f$ is the identity mapping on $U$.\\[1mm]
(4) For any unary polynomial $f$ such that $f(\beta)\not\sse\al$ there is an
$(\al,\beta)$-minimal set $U$ such that $f$ witnesses that $U$ and $f(U)$
are polynomially isomorphic.
\end{lemma}

Minimal sets of a Mal'tsev algebra form a particularly dense collection.

\begin{lemma}[Folklore]\label{lem:Maltsev-minimal-sets}
Let $\zA$ be a finite Mal'tsev algebra and $\al\prec\beta$ for 
$\al,\beta\in\Con(\zA)$. Then for any $a,b\in\zA$ with $(a,b)\in\beta-\al$,
there is an $(\al,\beta)$-minimal set $U$ such that $a^\al\cap U\ne\eps$ 
and $b^\al\cap U\ne\eps$.
\end{lemma}

%%%%%%%%%%%%%%%%%%%%%%%%%%%%%%%%%%%
\subsection{Semilattice block Mal'tsev algebras}\label{sec:block}

Since the fewer basic operations an algebra has, the richer the corresponding 
constraint language, we assume that the algebras we are dealing with have 
only two basic operations, just enough to guarantee the required properties. 
Therefore we assume that our semilattice block Mal'tsev 
algebras have only two basic operations: a binary operation $\cdot$ that 
we will often omit, and a ternary operation $\mal$ satisfying the conditions 
specified earlier. 
For elements $a,b\in\zA$ such that $ab=ba=b$ we write $a\le b$. 

\begin{lemma}\label{lem:dot-inequality}
Let $\zA$ be an SBM algebra. By choosing a reduct of $\zA$ we may assume 
that\\[1mm]
Operation $\cdot$ satisfies the equation $x(xy)=xy$; and for any 
$a,b\in\zA$, $a\le ab$.\\[1mm]
Operation $\mal$ can be chosen such that 
for any $a,b,c\in\zA$, $\mal(a,b,c)^{\sg_\zA}=(abc)^{\sg_\zA}$. 
\end{lemma}

\begin{proof}
(1) Follows from Proposition~10 of \cite{Bulatov16:connectivity} .

(2) Consider the operation $\mal'(x,y,z)=\mal(x,y,z)xyz$. If $B$ is a 
$\sg_\zA$-block, then, since $ab=a$ for any $a,b\in B$, operation
$\mal'$ is Mal'tsev on $B$. Also, as $\zA\fac{\sg_\zA}$ is term equivalent to
a semilattice, $d=\mal(a,b,c)^{\sg_\zA}$ belongs to the subsemilattice 
of  $\zA\fac{\sg_\zA}$ generated by $a^{\sg_\zA},b^{\sg_\zA},c^{\sg_\zA}$.
Therefore $\mal'(a,b,c)^{\sg_\zA}=d(abc)^{\sg_\zA}=(abc)^{\sg_\zA}$, 
and we can choose $\mal'$ for $\mal$.
\end{proof}

Next we show some useful properties of SBM algebras. 
Let $\zA$ be an SBM algebra and $\max(\zA)$ the maximal block
of $\sg$, that is, $\max(\zA)\cdot a\sse\max(\zA)$ for all $a\in\zA$. 

\begin{lemma}\label{lem:irreducible-maximal}
(1) The equivalence relation $\th_\zA$ whose blocks are $\max(\zA)$ and all the 
remaining elements form singleton blocks, is a congruence. \\
(2) Let $\rel$ be a subdirect product of SBM algebras $\vc\zA n$ and the 
equivalence relation $\th_\rel$ is such that its blocks are 
$\max(\rel)=\rel\cap(\max(\zA_1)\tm\dots\tm\max(\zA_n))$, 
and all the remaining elements form singleton blocks. Then $\th_\rel$ is a 
congruence.
\end{lemma}

\begin{proof}
(1) It suffices to observe that for any $a\in\max(\zA)$ we have $ax,xa$, 
$\mal(a,x,y),\mal(x,a,y),\mal(x,y,a)\in\max(\zA)$ for any $x,y$, and 
therefore all non-constant polynomials of $\zA$ preserve $\max(\zA)$.

(2) is similar to (1).
\end{proof}

When dealing with a relation over algebras $\vc\zA n$ or a CSP with domains
$\zA_v$ we will simplify the notation $\th_{\zA_i},\th_{\zA_v}$ to $\th_i,\th_v$.

\begin{lemma}\label{lem:min-set}
Every $(\al,\beta)$-minimal set, for $\al\prec\beta\le\th_\zA$, is a subset of $\max(\zA)$.
\end{lemma}

\begin{proof}
Let $U$ be a $(\al,\beta)$-minimal set and $f$ an idempotent polynomial 
with $f(\zA)=U$ and $f(\beta)\not\sse\al$. Since $\beta\le\th_\zA$,  
$c,d\in U\cap\max(\zA)$ for some $(c,d)\in\beta-\al$, as otherwise we would 
have $f(\beta)\sse\al$. 
Take $a\in\max(\zA)$ and set $g(x)=f(x)a$. For any $b\in U\cap\max(\zA)$
we have $g(b)=f(b)a=ba=b$.
Therefore $g(\beta)\not\sse\al$ and $g(\zA)\sse \max(\zA)$.
Finally, $f(\max(\zA))\sse\max(\zA)$, therefore 
$f\circ g(\zA)\sse U\cap\max(\zA)$
and $f\circ g(x)=x$ for $x\in U\cap\max(\zA)$. As $U$ is minimal, 
$U=U\cap\max(\zA)$.
\end{proof}

%%%%%%%%%%%%%%%%%%%%%%%%%%%%%%%%%%%%%%%%%%
\subsection{Maroti's reduction}\label{sec:maroti}

In this section we describe a reduction introduced by Maroti in \cite{Maroti10:tree} that allows us to reduce CSPs over SBM algebras to CSPs over SBM algebras of a certain restricted type. More precisely, it allows us to assume that every domain $\zA$ is either a Mal'tsev algebra with $\mal$ as a Mal'tsev operation, or it contains a \emph{minimal element} $a$, that is, an element such that $ab=ba=b$ for all $b\in\zA$. Moreover, as is easily seen, such element is unique and forms a $\sg_\zA$-block, which is also the smallest element of the semilattice $\zA\fac{\sg_\zA}$.

Let $f$ be an idempotent unary polynomial of algebra $\zA$ and $A$ the universe of $\zA$. The \emph{retract} $f(\zA)$ of $\zA$ is the algebra with universe $f(A)$, whose basic operations are of the form $f\circ g$, given by $f\circ g(\vc xn)=f(g(\vc xn))$ for $\vc xn\in f(A)$, where $g$ is a basic operation of $\zA$. 

\begin{lemma}\label{lem:retract}
A retract of an SBM algebra through an idempotent polynomial is an SBM algebra.
\end{lemma}

\begin{proof}
Let $f$ be an idempotent polynomial.
Let $g_1(x,y)=f(xy)$, $\mal_1(x,y,z)=f(\mal(x,y,z))$ be the basic operations of the  
retract, $\zA_1=f(\zA)$, and $\sg_1=\sg_\zA\red{\zA_1}$.
Firstly, note that $\sg_1$ is a congruence of $\zA_1$ and $\zA_1$ is an idempotent algebra.
Since $\zA\fac{\sg_\zA}$ is term equivalent to a semilattice and any retract of a semilattice by a semilattice polynomial 
is a semilattice, so is $\zA_1\fac{\sg_1}$. Finally,
\begin{eqnarray*}
&& \mal_1(x,y,y)=f(\mal(x,y,y))=f(x)=x\\
&& \mal_1(y,y,x)=f(\mal(y,y,x))=f(x)=x,
\end{eqnarray*}
for any $x,y\in\zA_1$ with $x\eqc{\sg_1}y$.
\end{proof}

The results of \cite{Maroti10:tree} imply the following.
Let $\gA$ be a class of similar finite algebras closed under subalgebras, 
and retracts via idempotent unary polynomials. 
Suppose that $\gA$ has a term operation $f$ satisfying the 
following conditions for some $\zB\in\gA$:
\begin{itemize}
\item[(1)]
$f(x,f(x,y))=f(x,y)$ for any $x,y\in\zB$;
\item[(2)]
for each $a\in\zB$ the mapping $x\mapsto f(a,x)$ is not surjective;
\item[(3)]
the set $C$ of $a\in\zB$ such that $x\mapsto f(x,a)$ is surjective 
generates a proper subalgebra of $\zB$.
\end{itemize}
Then $\CSP(\gA)$ is polynomial time reducible to $\CSP(\gA-\{\zB\})$.

By Lemma~\ref{lem:dot-inequality} the operation $\cdot$ of the class of SBM algebras from 
$\gA$ satisfies condition (1). If the operation $a\cdot x$ 
is surjective for some $a$, then $a\le x$  for all $x\in\zB$. Therefore the only case 
when condition (2) is not satisfied is when $\zB$ has a minimal element. 
Finally, condition (3) is satisfied whenever $\zB$ is not a Mal'tsev algebra. 
Therefore, choosing $\zB$ to be a maximal (in terms of cardinality) 
algebra from $\gA$ satisfying conditions (1)--(3) we may only consider 
instances of $\CSP(\gA)$, in which every domain has a minimal element
or is a Mal'tsev algebra.

\begin{corollary}\label{cor:maroti}
Every instance $\cP\in\CSP(\gA)$ can be reduced in polynomial time
to polynomially many instances over algebras each of which either is 
Mal'tsev or has a  minimal element.
\end{corollary}

Throughout the rest of the paper $\gA$ is a finite class of finite SBM algebras 
closed under taking subalgebras, quotient algebras, and retracts through
unary idempotent polynomials.

%%%%%%%%%%%%%%%%%%%%%%%%%%%%%%%%%%
%%%%%%%%%%%%%%%%%%%%%%%%%%%%%%%%%%
\section{Separating congruences}\label{sec:coherent}

In this section we develop a method that will lead to some way to decompose
CSPs over SBM algebras. First, we introduce and study the
notion of separation of prime intervals.
Let $\rel$ be a subdirect product of $\zA_1\tm\dots\tm\zA_n$ and 
$\al,\beta\in\Con(\zA_i)$, $\gm,\dl\in\Con(\zA_j)$,  for some $i,j\in[n]$, 
such that $\al\prec\beta$, $\gm\prec\dl$. Recall that interval $(\al,\beta)$ 
can be 
separated from $(\gm,\dl)$ if there is a unary polynomial $\bbf$ of $\rel$ such 
that $f_i(\beta)\not\sse\al$ while $f_j(\dl)\sse\gm$. If $\bbf$ satisfies this 
property we will also say that $\bbf$ {\em separates} $(\al,\beta)$ from 
$(\gm,\dl)$. In the definition above it is possible that $i=j$ or that $n=1$; 
in this cases the argument in some proofs may be slightly different. To avoid 
such complications we will always assume that $i\ne j$, as the following 
lemma allows us to do. 

\begin{lemma}\label{lem:separation-separation}
Let $\relo$ be the binary equality relation on $\zA$. Prime interval 
$(\al,\beta)$, $\al\prec\beta\le\th_\zA$, can be separated from 
$(\gm,\dl)$, $\gm\prec\dl\le\th_\zA$, 
as intervals in $\Con(\zA)$ if and only if $(\al,\beta)$ can be separated 
from $(\gm,\dl)$ in $\relo$ (as intervals in the congruence lattices
of the factors of a binary relation). 
\end{lemma}

\begin{proof}
Note that for any polynomial $\bbf$ of $\relo$ its action on the first and second 
factors of $\relo$ is the same polynomial of $\zA$. By definition 
$\al\prec\beta$ can be 
separated from $\gm\prec\dl$ in $\Con(\zA)$ if and only if there is 
a unary polynomial $f$ of $\zA$, $f(\beta)\not\sse\al$ while $f(\dl)\sse\gm$. 
This condition can be expressed as follows: there is 
a unary polynomial $\bbf$ of $\relo$, $f_1(\beta)\not\sse\al$ while 
$f_2(\dl)\sse\gm$, which precisely means that $(\al,\beta)$ can be separated 
from $(\gm,\dl)$ in $\relo$
\end{proof}

In Section~\ref{sec:separation}
we study the sets of intervals that cannot be separated from each other. 
These sets will later give us some sort of decomposition of CSP instances.
Collapsing polynomials introduced in Section~\ref{sec:collapsing} yeild  
one of the main ingredients of the solution algorithm. 
Section~\ref{sec:splits} provides a sufficient condition for separation of 
intervals and a related notion of decomposition, which is the second ingredient.

%%%%%%%%%%%%%%%%%%%%%%%%%%%%%%%%%%
\subsection{Basic properties of separation}\label{sec:separation}

Let again $\rel$ be a subdirect product of SBM algebras 
$\zA_1\tm\ldots\tm \zA_n$, $i,j\in[n]$, and $\al,\beta\in\Con(\zA_i)$, 
$\gm,\dl\in\Con(\zA_j)$ with $\al\prec\beta\le\th_i$, $\gm\prec\dl\le\th_j$. 

First, we show that separating polynomials can be chosen to satisfy certain 
simple conditions.

\begin{lemma}\label{lem:max-separate}
If $(\al,\beta)$ can be separated from $(\gm,\dl)$ then there is a polynomial 
$\bbf$ that separates $(\al,\beta)$ from $(\gm,\dl)$ and such that 
$f_\ell(\zA_\ell)\sse\max(\zA_\ell)$ for every $\ell\in[n]$.
\end{lemma}

\begin{proof}
Let $\bg$ separate $(\al,\beta)$ from $(\gm,\dl)$. Choose a tuple 
$\ba\in\max(\rel)$ and consider the 
polynomial $\bbf(x)=\bg(x)\cdot\ba$. As is easily seen, $f_\ell(\zA_\ell)\sse
\max(\zA_\ell)$ for $\ell\in[n]$. Since $g_j(\dl)\sse\gm$, we have 
$f_j(\dl)\sse\gm$.
Finally, take $a,b\in\max(\zA_i)\cap g_i(\zA_i)$ with $(a,b)\in\beta-\al$ 
and $a',b'\in\max(\zA_i)$
such that $g_i(a')=a$, $g_i(b')=b$. By Lemma~\ref{lem:minimal-sets}(4) 
and Lemma~\ref{lem:min-set} such elements exist, because
$g_i(\beta)\not\sse\al$ and all the nontrivial (that is, different from an
$\al$-block) $\beta_i$-blocks are inside $\max(\zA_i)$. Then 
$$
f_i(a')=g_i(a')\ba[i]=a\ba[i]=a \ne b=b\ba[i]=g_i(b')\ba[i]=f_i(b').
$$
\end{proof}

From now on we assume that all polynomials separating intervals
satisfy the conditions of Lemma~\ref{lem:max-separate}.

\begin{lemma}\label{lem:selected}
If $(\al,\beta)$ can be separated from $(\gm,\dl)$ then, for any
$(\al,\beta)$-minimal set $U$, there is an idempotent unary
polynomial $\bg$ such that $g_i(\zA_i)=U$, and $\bg$ separates 
$(\al,\beta)$ from $(\gm,\dl)$.
\end{lemma}

\begin{proof}
Let $\bbf$ separate $(\al,\beta)$ from $(\gm,\dl)$. Then by 
Lemma~\ref{lem:minimal-sets}(4) $f_i(\zA_i)$ contains an
$(\al,\beta)$-minimal set $U'$, and there is an idempotent polynomial $h_i$
with $h_i(\zA_i)=U'$. The polynomial
$h_i$ can be extended to a polynomial $\bh$ of $\rel$. Then $\bbf'=\bh\circ\bbf$
separates $(\al,\beta)$ from $(\gm,\dl)$ and $f'_i(\zA_i)=U'$. 

By Lemma~\ref{lem:minimal-sets}(2) there is an $(\al,\beta)$-minimal 
set $U''$ with $f'_i(U'')=U'$ and an
idempotent polynomial $h'_i$ with $h'_i(U')=U''$. As above, the polynomial $h'_i$
can be extended to a polynomial $\bh'$ of $\rel$. For a certain $k$,
$(\bbf'\circ\bh')^k$ is idempotent, separates $i$ from $j$, and
$(f'_i\circ h'_i)^k(\zA_i)=U''$. 
Now the lemma follows easily from Lemma~\ref{lem:minimal-sets}(1).
\end{proof}

Let $\cI_\rel$ be the set of triples $(i,\al,\beta)$ such that $i\in[n]$, 
$\al,\beta\in\Con(\zA_i)$ and $\al\prec\beta\le\th_i$. The relation 
`cannot be separated in $\rel$' on $\cI_\rel$ is clearly reflexive and transitive. 
Now, we prove it is also symmetric

\begin{lemma}\label{lem:symmetric}
If $(\al,\beta)$ can be separated from $(\gm,\dl)$ then $(\gm,\dl)$ can be 
separated from~$(\al,\beta)$.
\end{lemma}

\begin{proof}
Let $\vc Uk$ be all the $(\al,\beta)$-minimal sets. By 
Lemma~\ref{lem:selected}, for every $U_\ell$, there is an idempotent 
unary polynomial $\bg^{(\ell)}$ separating $(\al,\beta)$ from $(\gm,\dl)$
and such that $g^{(\ell)}_i(\zA_i)=U_\ell$.  Take a 
$\dl$-block $B$ that contains more than one $\gm$-block, 
a tuple $\ba\in\rel$ such that $\ba[j]\in B$, and set
$\ba^{(\ell)}=\bg^{(\ell)}(\ba)$. By Lemmas~\ref{lem:min-set} 
and~\ref{lem:max-separate}
$\ba^{(1)}\zd\ba^{(k)}\in\max(\rel)$ 
and $\vc Uk\sse\max(\zA_i)$, and $B\sse\max(\zA_j)$. The operation 
$\bh^{(\ell)}(x)=\mal(x,\bg^{(\ell)}(x),\ba^{(\ell)})$ satisfies the conditions 
\begin{itemize}
\item
$h^{(\ell)}_i(x)=\mal(x,g^{(\ell)}_i(x),\ba^{(\ell)}[i])=
\mal(x,x,\ba^{(\ell)}[i]) =\ba^{(\ell)}[i]$ for all $x\in U_\ell$; 
\item
$h^{(\ell)}_j(x)=\mal(x,g^{(\ell)}_j(x),\ba^{(\ell)}[j])\eqc{\al_j}
\mal(x,\ba^{(\ell)}[j],\ba^{(\ell)}[j]) =x$ for all $x\in B$; 
\item
$\bh^{(\ell)}(\rel)\sse\max(\rel)$.
\end{itemize}
We are going to compose the polynomials $\bh^{(\ell)}$ such that the
composition collapses $\beta$. To this end take a sequence
$1=\ell_1,\ell_2,\ldots$ such that $U_{\ell_2}$ is a subset of the range of
$\ov h^{(1)}=h_i^{(\ell_1)}$, and, for $s>2$, $U_{\ell_s}$ is a subset of
the range of 
$\ov h^{(s-1)}=h_i^{(\ell_{s-1})}\circ\ldots\circ h_i^{(\ell_1)}$. Since
$|\ov h^{(s)}(\zA_i)|<|\ov h^{(s-1)}(\zA_i)|$, there is $r$ such that
$|\ov h^{(r)}(\zA_i)|$ contains no $(\al,\beta)$-minimal
sets. Therefore, setting
$\bh(x)=\bh^{(\ell_r)}(\bh^{(\ell_{r-1})}(\ldots\bh^{(\ell_1)}(x)\ldots))$ we
have that $h_i$ collapses all the $(\al,\beta)$-minimal sets, and $h_j$ acts
identically on $B\fac{\al_j}$. Thus, $\bh$ separates $(\gm,\dl)$ from 
$(\al,\beta)$.
\end{proof}

Lemma~\ref{lem:symmetric} together with the observation before it 
shows that the relation `cannot be separated' is an equivalence relation 
on $\cI$. 
%% We call its classes {\em coherent sets} with respect to $\ov\al,\ov\beta$.

%%%%%%%%%%%%%%%%%%%%%%%%%%%%%%%%%%%%%%%%
\subsection{Collapsing polynomials}\label{sec:collapsing}

Intuitively, a collapsing polynomial for some prime interval $\al\prec\beta$
in an algebra or a subdirect product of algebras is a polynomial that collapses 
all prime intervals that can be separated from $\al\prec\beta$ and only
such prime intervals. 

Let $\rel$ be a subdirect product of SBM algebras $\zA_1\tms\zA_n$, 
and $(i,\al,\beta)\in\cI_\rel$. 
A unary idempotent polynomial $\bbf$ of $\rel$ is called 
\emph{$(\al,\beta)$-collapsing} if the following conditions hold:
\begin{itemize}
\item[(C1)]
for any $(j,\gm,\dl)\in\cI_\rel$, it holds $f_j(\dl)\sse\gm$, unless 
$(\al,\beta)$ and $(\gm,\dl)$ cannot be separated;
\item[(C2)]
for any $(j,\gm,\dl)\in\cI_\rel$ such that $(\al,\beta)$, $(\gm,\dl)$ cannot be 
separated, the set $f_j(\zA_j)$ is a $(\gm,\dl)$-minimal set.
\end{itemize}

First, we show that $(\al,\beta)$-collapsing polynomials exist even if we impose 
some additional requirements.

\begin{lemma}\label{lem:relation-collapse}
Let $\rel$ be a subdirect product of SBM algebras $\zA_1\tms\zA_n$ 
and $(i,\al,\beta)\in\cI_\rel$, 
and let $\ba\in\rel$ be such that $\ba[i]$ belongs to a $\beta$-block 
containing more than one $\al$-block and $b\in\zA_i$ with 
$(\ba[i],b)\in\beta-\al$. Then there is an $(\al,\beta)$-collapsing polynomial 
$\bbf$ of $\rel$ such that $\bbf(\ba)=\ba$ and $f_i(b)\eqc\al b$.
\end{lemma}

\begin{proof}
First, we find an $(\al,\beta)$-collapsible polynomial. For every 
$(j,\gm,\dl)\in\cI_\rel$ such that $(\al,\beta)$
can be separated from $(\gm,\dl)$ there is an idempotent polynomial 
$\bg^{j\gm\dl}$ such that $g^{j\gm\dl}_j(\dl)\sse\gm$, but 
$g^{j\gm\dl}_i(\beta)\not\sse\al$. Moreover, we may assume by 
Lemma~\ref{lem:selected} that for every $\bg^{j\gm\dl}$, 
$g^{j\gm\dl}_i(\zA_i)=U$ for the same $(\al,\beta)$-minimal set $U$.
Composing all such polynomials we obtain a polynomial $\bh$ such that
$h_i(\zA_i)=U$, and so $h_i(\beta)\not\sse\al$, and $h_j(\dl)\sse\gm$
for any $j,\gm,\dl$ as above. By iterating $\bh$ can be assumed idempotent. 
Choose $\bh$ to have the smallest image among unary idempotent polynomials 
such that $h_i(\zA_i)$ is an $(\al,\beta)$-minimal set and $h_j(\dl)\sse\gm$ 
for any $(j,\gm,\dl)\in\cI_\rel$ 
such that $(\al,\beta)$ can be separated from $(\gm,\dl)$.

Suppose now that for some $(j,\gm,\dl)\in\cI_\rel$ such that the interval 
$(\al,\beta)$ cannot be
separated from $(\gm,\dl)$ the set $U'=h_j(\zA_j)$ is not a 
$(\gm,\dl)$-minimal set. Then, since $h_j(\dl)\not\sse\gm$, the set
$U'$ contains a $(\gm,\dl)$-minimal set $U''$ by 
Lemma~\ref{lem:minimal-sets}(4). Let $g$ be an idempotent polynomial
of $\zA_j$ with $g(\zA_j)=U''$ and $\bg$ its extension to a polynomial
of $\rel$. Then $\bh'=\bg\circ\bh$ satisfies the following conditions: \\[1mm]
-- $h'_j(\zA_j)=U''$ and $h'_j(\dl)\not\sse\gm$;\\[1mm]
-- $h'_i(\beta)\not\sse\al$, because $(\al,\beta)$ cannot be separated from 
$(\gm,\dl)$;\\[1mm]
-- $|\bh'(\rel)|<|\bh(\rel)|$.\\[1mm]
Iterating $\bh'$ it can be assumed idempotent. Then the last property contradicts
the choice of $\bh$. Therefore $\bh$ is $(\al,\beta)$-collapsing.

Let $\al_i=\al,\beta_i=\beta$, and for $j\in[n]-\{i\}$ let 
$\al_j=\beta_j=\th_j$. 
It is not hard to see that $\ov\al\preceq\ov\beta$. Indeed, suppose 
$\eta\in\Con(\rel)$ is such that $\ov\al<\eta\le\ov\beta$ and let $i=n$.
Then there are $(\bc,c),(\bd,d)\in\rel$ such that $\ang{(\bc,c),(\bd,d)}\in\eta$
such that $\ang{\bc,\bd}\in\ov\al_{[n-1]}$ and $\ang{c,d}\in\beta-\al$.
We show that for any $\ang{(\bc',c'),(\bd',d')}\in\ov\beta$ we have  
$\ang{(\bc',c'),(\bd',d')}\in\eta$. In fact, by Proposition~\ref{pro:basics} 
it suffices to show that 
$\ang{(\bc'',c''),(\bd'',d'')}\in\eta$ for some $\bc'',\bd''\in\pr_{[n-1]}\rel'$
where $\rel'=\max(\rel)$ and $\ang{\bc'',\bd''}\in\ov\al_{[n-1]}$,
and some $c'',d''\in\max(\zA_n)$ with $c''\eqc\al c'$, $d''\eqc\al d'$. 
Since $\rel'$ is a Mal'tsev algebra by Lemma~\ref{lem:Maltsev-minimal-sets} 
applied to conguences $\al\prec\beta$ there is a polynomial $\bbf$ of 
$\rel$ such that $c''=f_n(c)\eqc\al c', d''=f_n(d)\eqc\al d'$ and 
$\bbf(\rel)\sse\rel'$. Let $\bc''=\bbf_{[n-1]}(\bc),\bd''=\bbf_{[n-1]}(\bd)$.
Then $\ang{(\bc'',c''),(\bd'',d'')}\in\eta$. Also, since $\beta\ne\al$, we have 
$\ov\al\prec\ov\beta$.

By Lemma~\ref{lem:Maltsev-minimal-sets} there is an $(\al,\beta)$-minimal
set $U$ such that $\ba[i]^\al\cap U,b^\al\cap U\ne\eps$. Moreover, 
an $(\al,\beta)$-collapsing polynomial $\bh$ can be chosen such that 
$h_i(\zA_i)=U$. Then set $\bbf(x)=\mal(\bh(x),\bh(\ba),\ba)$.
For the polynomial $\bbf$ we have:\\[1mm]
-- $\bbf(\ba)=\mal(\bh(\ba),\bh(\ba),\ba)=\ba$;\\[1mm]
-- $c=f_i(b)=\mal(h_i(b),h_i(\ba[i]),\ba[i])\eqc\al
\mal(h_i(b),\ba[i],\ba[i])=h_i(b)\eqc\al b$, because, since $\bh$ is idempotent,
$h_i(\ba[i])\eqc\al\ba[i]$ and $h_i(b)\eqc\al b$;\\[1mm]
-- for any $(j,\gm,\dl)\in\cI_\rel$ such that and $(\al,\beta),(\gm,\dl)$ 
can be separated, $f_j(\dl)\sse\gm$.\\[1mm]
By iterating $\bbf$ we obtain an idempotent polynomial $\bbf'$ that satisfies
all the conditions above. Indeed, the first and third conditions are 
straightforward, while the second one follows from the equality
$f_i(c)\eqc\al c$. Finally, for any $(j,\gm,\dl)\in\cI_\rel$ such 
that $(\al,\beta),(\gm,\dl)$ cannot
be separated we have $f'_j(\dl)\not\sse\gm$, because $f'_i(\beta)\not\sse\al$.
Also, $f'_j(\zA_j)$ is a $(\gm,\dl)$-minimal set, because $h_j(\zA_j)$ 
is a one. 

Thus, $\bbf'$ satisfies all the required conditions. The lemma is proved.
\end{proof}

%%%%%%%%%%%%%%%%%%%%%%%%%%%%%%%%
\subsection{Splits and alignments}\label{sec:splits}

In this section we present a sufficient condition for two prime intervals to be 
separated. As we shall see using this condition certain projections of a relation 
can be partitioned into a small number of subdirect products of smaller algebras.

Let $\rel$ be a subdirect product of $\zA_1\tms\zA_n$, 
$\al_i,\beta_i\in\Con(\zA_i)$, $i\in[n]$, such that 
$\al_i\prec\beta_i\le\th_{\zA_i}$. An element $a\in\zA_i$,
$i\in[n]$, is called \emph{$\al_i\beta_i$-split} if there is a
$\beta_i$-block $B$ and $b,c\in B$ such that $ab\not\eqc{\al_i} ac$. Note that 
no element from $\max(\zA_i)$ is $\al_i\beta_i$-split, while the
minimal element is $\al_i\beta_i$-split.
We say that $i,j\in[n]$ are not \emph{$\ov\al\ov\beta$-aligned}
if there is $\ba\in\rel$ such that $\ba[i]$ is not $\al_i\beta_i$-split and
$\ba[j]$ is $\al_j\beta_j$-split, or the other way round.

\begin{lemma}\label{lem:align-coherence}
If $i,j$ are not $\ov\al\ov\beta$-aligned then $(\al_i,\beta_i)$ can be 
separated from $(\al_j,\beta_j)$. 
\end{lemma}

\begin{proof}
It suffices to consider the case $n=2$, $i=1$, $j=2$. Let
$(a,b)\in\rel$ be such that $a$ is $\al_i\beta_i$-split, while $b$ is not 
$\al_j\beta_j$-split. Let also $(c,d)\in\rel'=\max(\rel)$.
Consider operation $\bbf((x_1,x_2))=(a,b)\cdot((x_1,x_2)\cdot(c,d))$. 
We claim that $f_1(\beta_1)\not\sse\al_1$ while $f_2(\beta_2)\sse\al_2$.

First, observe that all the values of the operation 
$g((x_1,x_2))=(x_1,x_2)\cdot(c,d)$ belong to $\max(\rel)$, and 
$g((x_1,x_2))=(x_1,x_2)$ for any $(x_1,x_2)\in\max(\rel)$. Then, for any  
$\beta_2$-block $B_2$ and any $a',b'\in B_2$ we have 
$f_2(a')=b(a'd)\eqc{\al_2}b(b'd)=f_2(b')$, as
$b$ is not $\al_2\beta_2$-split. Thus $f_2(\beta_2)\sse\al_2$. On the 
other hand, since $a$ is $\al_1\beta_1$-split, there is a $\beta_1$-block 
$B_1$ and $a'',b''\in B_1$ such that 
$f_1(a'')=a(a''c)=aa''\not\eqc{\al_1} ab''=a(b''c)=f_1(b'')$. The second and the 
second last equalities hold because, as $\beta_1\sse\th_1$ and 
$B_1$ is a nontrivial $\beta_1$-block, we have $B_1\sse\max(\zA_1)$.
Therefore $f_1(\beta_1)\not\sse\al_1$.
\end{proof}

%%%%%%%%%%%%%%%%%%%%%%%%%%%%%%%%%%
%%%%%%%%%%%%%%%%%%%%%%%%%%%%%%%%%%
\section{From relations to instances}\label{sec:to-instances}

Here we apply the results of the previous section to CSP instances. 
In particular, we introduce coherent sets of an instance and show that
if an instance has solutions on every coherent set, which are consistent
in some weak sense, then the entire instance has a solution.

Let $\cP=(V,\cA,\cC)$ be a 3-minimal instance of $\CSP(\gA)$. We assume that
the domain $\zA_v$ of each variable $v\in V$ is the set of solutions $\cS_v$, 
and so the constraint relations are subdirect products of the domains.

Since separation of prime intervals depends only on binary projections of a 
relation, it can be defined for 3-minimal instances as well. More precisely, 
let $\cI_\cP$ (or just $\cI$ if $\cP$ is clear from the context) be the set of all 
triples $(v,\al,\beta)$, where $v\in V$,
$\al,\beta\in\Con(\zA_v)$ are such that $\al\prec\beta\le\th_v$.
Let $(v,\al,\beta),(w,\gm,\dl)\in\cI$; we say that $(\al,\beta)$ cannot separated 
from $(\gm,\dl)$ if this is the case for 
$\cS_{vw}$. Due to 3-minimality --- we can consider sets of solutions on 3 
variables --- this relation is transitive. It is also reflexive and symmetric
by Lemma~\ref{lem:symmetric}. 

Next we define two partitions of a CSP instance $\cP$. The first one, 
link partition allows us to reduce solving subinstances of $\cP$ to instances 
over smaller domains. The second one provides a sufficient condition
to have a link partition and is defined through alignment properties.

Let again $\cP=(V,\cA,\cC)$ be a 3-minimal instance of $\CSP(\gA)$. Partitions 
$A_{v1}\cup\ldots\cup A_{vk_v}=\zA_v$ for $v\in V$ are called a 
\emph{link partition} if the following condition holds:
\begin{itemize}
\item
For any $v,w\in V$, $k_v=k_w$, 
and there is a bijection $\vf_{vw}:[k_v]\to[k_w]$ such that for any 
$(a,b)\in\cS_{vw}$ and any $j\in[k_v]$, $a\in A_{vj}$ if and only if 
$b\in A_{w\vf_{vw}(j)}$.
\end{itemize}

Observe that, since $\cP$ is 3-minimal, the mappings $\vf_{vw}$ are 
consistent, that is, for any $u,v,w\in V$ it holds that 
$\vf_{vw}\circ\vf_{uv}=\vf_{uw}$. Without loss of generality we will 
assume that $\vf_{vw}$ is an identity mapping. 

As is easily seen the partition $A_{v1}\cup\ldots\cup A_{vk_v}=\zA_v$
defines a congruence of $\zA_v$. In particular, each of $A_{vi}$ is a 
subalgebra of $\zA_v$.

Let $\al_v,\beta_v\in\Con(\zA_v)$ for $v\in V$ be such that 
$\al_v\prec\beta_v\le\th_v$.
Variables $v,w\in V$ are \emph{$\ov\al\ov\beta$-aligned} if they are 
$\ov\al\ov\beta$-aligned in $\cS_{vw}$. 
In the following lemma we assume that every domain $\zA_v$ of $\cP$
either has a minimal element, or $\sg_{\zA_v}$ is the full congruence,
and so $\zA_v$ is a Mal'tsev algebra.

\begin{lemma}\label{lem:coherent-link-partition}
(1) If variables $v,w\in V$ of an instance $\cP=(V,\cA,\cC)$ are 
$\ov\al\ov\beta$-aligned and $\zA_v$ has a minimal element then 
$\zA_w$ also has a minimal element.\\
(2) If every domain of an instance $\cP=(V,\cA,\cC)$ has a minimal element 
and any two variables $v,w\in V$ are $\ov\al\ov\beta$-aligned, then $\cP$
has a link partition.
\end{lemma}

\begin{proof}
For every $v\in V$ let $L_v$ denote the set of $\al_v\beta_v$-split elements of 
$\zA_v$ and let $N_v$ denote the set of $\al_v\beta_v$-non-split elements. 
As we observed before Lemma~\ref{lem:align-coherence}, both sets are 
nonempty if $\zA_v$ has a minimal element, and $L_v=\eps$ if $\zA_v$ is
a Mal'tsev algebra.

(1) If $\zA_w$ is a Mal'tsev algebra then $v,w$ cannot be 
$\ov\al\ov\beta$-aligned since $L_w=\eps$, while $L_v,N_v\ne\eps$,
and $\cS_{vw}$ is a subdirect product.

(2) For any $v,w\in V$ and any pair $(a,b)\in\cS_{vw}$, 
$a\in L_v$ if and only if $b\in L_w$. Therefore $\cS_{vw}$ is
link-partitioned, as well as $\rel$ for any constraint
$C=\ang{\bs,\rel}\in\cC$.
\end{proof}

%%%%%%%%%%%%%%%%%%%%%%%%%%%%%%%%%%%
%%%%%%%%%%%%%%%%%%%%%%%%%%%%%%%%%%%
\section{The algorithm}\label{sec:algorithm}

In the first part of this section we introduce the property of block-minimality, 
the key property of CSP instances for our algorithm. We also prove
that block-minimality can be efficiently established. Then in the second part
we show that block-minimality is sufficient for the existence of a solution,
Theorem~\ref{the:block-minimal}, which is the main result of this section,
and provides a polynomial time algorithm for CSPs over SBM algebras.

%%%%%%%%%%%%%%%%%%%%%%%%%%%%%%%%%%%
\subsection{Block-minimality}\label{sec:block-minimality}

Let $\cP=(V,\cA,\cC)$ be a 3-minimal instance such that for every its
domain $\zA_v$ either $\sg_{\zA_v}$ is the full congruence, and so 
$\zA_v$ is a Mal'tsev algebra with Mal'tsev operation $\mal$, or $\zA_v$ has
a minimal element. 

Recall that $\cI_\cP$ or just $\cI$ denotes the set of all triples $(v,\al,\beta)$, 
where $v\in V$,
$\al,\beta\in\Con(\zA_v)$ are such that $\al\prec\beta\le\th_v$.
For a triple $(v,\al,\beta)\in\cI$ by $\cI(v,\al,\beta)$ we denote 
the set of all triples $(w,\gm,\dl)\in\cI$ such that $(\al,\beta)$ 
cannot be separated from $(\gm,\dl)$. Also, by $W_{v\al\beta}$ we
denote the set $\{w\mid (w,\gm,\dl)\in\cI(v,\al,\beta)\}$. Sets of
the form $W_{v\al\beta}$ are called \emph{coherent sets}.

Instance $\cP$ is said to be \emph{block-minimal} if for any 
$(v,\al,\beta)\in\cI$ the instance $\cP_{W_{v\al\beta}}$ is minimal.

In the next section we prove, Theorem~\ref{the:block-minimal}, that 
every block-minimal instance has a solution. To show that 
Theorem~\ref{the:block-minimal} gives rise to a polynomial-time 
algorithm for $\CSP(\gA)$ we need to show how block-minimality can be
established. We prove that establishing block-minimality can be reduced 
to solving polynomially many smaller instances of $\CSP(\gA)$.

\begin{prop}\label{pro:to-block-minimal}
Transforming an instance $\cP=(V,\cA,\cC)\in\CSP(\gA)$ to a block-minimal 
instance can be reduced to
solving polynomially many instances $\cP'=(V',\cA',\cC')\in\CSP(\gA)$ such 
that $V'\sse V$ and either $\zA'_v$ is a Mal'tsev algebra for all $v\in V'$, or
$|\zA'_v|<|\zA_v|$ for all $v\in V'$.
\end{prop}

Since the cardinalities of algebras in $\gA$ are bounded, the depth of recursion
when establishing block-minimality is also bounded. Therefore,  together with
Theorem~\ref{the:block-minimal} this proposition gives a polynomial 
time algorithm for $\CSP(\gA)$.

\begin{proof}
Using the standard propagation algorithm and Maroti's reduction 
(Section~\ref{sec:maroti}) we may assume that $\cP$ is 3-minimal and every 
$\zA_v$ is either Mal'tsev or has a minimal element. Take $(v,\al,\beta)\in\cI$ 
as in the definition of block-minimality. We need to show how to make
problems $\cP_{W_{v\al\beta}}$ minimal. If every $\zA_w$ for 
$w\in W_{v\al\beta}$ is Mal'tsev, $\cP_{W_{v\al\beta}}$ can be made 
minimal using the algorithm from \cite{Bulatov06:simple}.
If $\zA_w$ has a minimal element for some $w\in W_{v\al\beta}$ then 
set $\al_v=\al,\beta_v=\beta$, and for each $w\in W_{v\al\beta}$ choose
$\al_w,\beta_w$ in such a way that $(w,\al_w,\beta_w)\in\cI(v,\al,\beta)$.
Then by Lemma~\ref{lem:coherent-link-partition} and~\ref{lem:align-coherence} 
$\cP_{W_{v\al\beta}}$ is link partitioned, that is,
it is a disjoint union of instances $\cP_1\cup\dots\cup\cP_m$, where 
$\cP_i=(W_{v\al\beta},\cA^i,\cC^i)$ are such that 
$\zA_w=\zA_w^1\cup\dots\cup\zA_w^m$
is a disjoint union. We then transform them to minimal instances separately.

If at any stage there is a tuple from a constraint relation that does not extend to
a solution of a certain subinstance, we tighten the original problem $\cP$ and
start all over again. Observing that the set tuples from a constraint relation that 
can be extended to a solution of the subinstance is a subalgebra, the resulting
instance belongs to $\CSP(\gA)$ as well.
\end{proof}

%%%%%%%%%%%%%%%%%%%%%%%%%%%%%%
\subsection{Block-minimality and solutions of the CSP}\label{sec:solutions}

We now prove that block-minimality is a sufficient condition to have a solution.

\begin{theorem}\label{the:block-minimal}
Every block-minimal instance $\cP\in\CSP(\gA)$ with nonempty 
constraint relations has a solution.
\end{theorem}

\begin{proof}
Let $\cP=(V,\cA,\cC)$ be a 3-minimal and block-minimal instance from 
$\CSP(\gA)$, and
such that every domain $\zA_v$ is either a Mal'tsev algebra or has a minimal 
element. We make use of the following construction. Let 
$\gm_v\in\Con(\zA_v)$, $\gm_v\le\th_v$ for $v\in V$. A collection of mappings 
$\cM=\{\vf_{v\al\beta}\mid (v,\al,\beta)\in\cI\}$ is called an 
\emph{$\ov\gm$-ensemble} for $\cP$ if 
\begin{itemize}
\item[(1)] 
for every $(v,\al,\beta)\in\cI$ the mapping $\vf_{v\al\beta}$ is a solution of 
$\cP_{W_{v\al\beta}}$; and 
\item[(2)]
for every $(v,\al,\beta),(w,\gm,\dl)\in\cI$, and any 
$u\in W_{v\al\beta}\cap W_{w\gm\dl}$, it holds 
$\vf_{v\al\beta}(u)\eqc{\gm_u}\vf_{w\gm\dl}(u)$;
\item[(3)]
for any $C=\ang{\bs,\rel}\in\cC$ the tuple $\ba$ where 
$\ba[u]=\vf_{v\al\beta}\fac{\gm_v}$ for $u\in\bs$ and any $(v,\al,\beta)\in\cI$ 
with $u\in W_{v\al\beta}$, belongs to $\rel\fac{\ov\gm_\bs}$.
\end{itemize}
We prove that for any $\gm_v\in\Con(\zA_v)$, $\gm_v\le\th_v$ for $v\in V$
the instance $\cP$ has a $\ov\gm$-ensemble. 

If $\gm_v=\th_v$ for each $v\in V$ then any collection of solutions
$\vf_{v\al\beta}$ of $\cP_{W_{v\al\beta}}$ such that 
$\vf_{v\al\beta}(u)\in\max(\zA_u)$ for all
$(v,\al,\beta)\in\cI$, and $u\in W_{v\al\beta}$, satisfies the conditions of a 
$\ov\gm$-ensemble. Moreover by the block-minimality of $\cP$ such solutions
exist. 

If $\gm_v=\zz_v$ for $v\in V$ then for any $(v,\al,\beta),(w,\gm,\dl)\in\cI$ 
condition (2) implies $\vf_{v\al\beta}(u)=\vf_{w\gm\dl}(u)$
for $u\in W_{v\al\beta}\cap W_{w\gm\dl}$. Let us denote this value by 
$\psi(u)$. Then condition (3) implies that $\psi$ is a solution of $\cP$.

Finally, the inductive step follows from Lemma~\ref{lem:inductive}.
\end{proof}

\begin{lemma}\label{lem:inductive}
Let $\cP=(V,\cA,\cC)\in\CSP(\gA)$ be a 3-minimal and block-minimal instance
such that every $\zA_v$, $v\in V$, either is Mal'tsev or has a minimal element.
Let $v\in V$ and $\beta_w,\gm_w\in\Con(\zA_w)$, $w\in V$, be such that
$\beta_w\preceq\gm_w\le\th_w$, $\beta_v\prec\gm_v$ 
and $\beta_w=\gm_w$ for $w\ne v$. If there is a $\ov\gm$-ensemble for $\cP$
then there is a $\ov\beta$-ensemble for $\cP$.
\end{lemma}

\begin{proof}
Let $\cM=\{\vf_{w\gm\dl}\mid (w,\gm,\dl)\in\cI\}$ be a $\ov\gm$-ensemble and
$\xi(u)=\vf_{w\gm\dl}(u)^{\gm_u}$ for $u\in W_{w\gm\dl}$. 
By condition (2) for 
$\ov\gm$-ensembles this definition is consistent. If $\xi(v)$ is a 
$\gm_v$-block that is equal to an $\beta_v$-block, then $\cM$ is also
a $\ov\beta$-ensemble, and there is nothing to prove.

Otherwise let $B$ be the $\beta_v$-block 
containing $\vf_{v\al\beta}(v)$. We show that for every 
$(w,\gm,\dl)\in\cI$ with $v\in W_{w\gm\dl}$ a solution $\vf'_{w\gm\dl}$ can be
found such that $\vf'_{w\gm\dl}(v)\in B$ and 
$\vf'_{w\gm\dl}(u)\eqc{\gm_u}\vf_{w\gm\dl}(u)$. Then, 
setting $\vf'_{w\gm\dl}=\vf_{w\gm\dl}$ for $(w,\gm,\dl)\in\cI$ such that 
$v\not\in W_{w\gm\dl}$ and $\cM'=\{\vf'_{w\gm\dl}\mid (w,\gm,\dl)\in\cI\}$ 
we conclude that $\cM'$ is a $\ov\beta$-ensemble.

Let $(w,\gm,\dl)\in\cI$ be such that $v\in W_{w\gm\dl}$, and let 
$W=W_{v\al\beta}$, $U=W_{w\gm\dl}$, $\vf=\vf_{v\al\beta}\red{W\cap U}$, 
$\psi=\vf_{w\gm\dl}$. Note that in this 
notation $\cS_W$, $\cS_U$, and $\cS_{W\cap U}$ are the sets of solutions 
of $\cP_{W_{v\al\beta}}$, $\cP_{W_{w\gm\dl}}$, and 
$\cP_{W_{v\al\beta}\cap W_{w\gm\dl}}$. It will 
often be convenient for us to treat these sets as relations rather than 
sets of solutions of a CSP. Then 
$\pr_{W\cap U}\cS_W,\pr_{W\cap U}\cS_U\sse\cS_{W\cap U}$, 
and so $\vf,\pr_{W\cap U}\psi\in\cS_{W\cap U}$. 

Let $\bbf$ be a $(\beta_v,\gm_v)$-collapsing polynomial of $\cS_U$. By 
Lemma~\ref{lem:relation-collapse} it can be selected such that 
$\psi\in\bbf(\cS_U)$
and $B\cap f_v(\zA_v)\ne\eps$. Let $\pi=\bbf_{W\cap U}(\vf)$. We show that 
the mapping $\vf'$ on $U$ given by $\vf'(u)=\pi(u)$ for $u\in W\cap U$, and
$\vf'(u)=\psi(u)$ for $u\in U-W$ is a solution from $\cS_U$. Since $\vf(v)\in B$
and  $B\cap f_v(\zA_v)\ne\eps$, that is, $f_v(B)\sse B$ as $\bbf$ is idempotent, 
we have $\pi(v)=f_v(\vf(v))\in B$. Also, as for every $u\in (W\cap U)-\{v\}$, 
we have 
$$
\vf'(u)=\pi(u)=f_u(\vf(u))\eqc{\beta_u}f_u(\psi(u))=\psi(u).
$$
Therefore, $\vf'$ satisfies condition (2) of $\ov\beta$-ensembles for $w,j$.

Now we prove that $\vf'$ is a solution from $\cS_U$. Let $C=\ang{\bs,\rel}$
be a constraint from $\cP_U$, $W'=\bs\cap W$ and $\ba=\pr_{W'}\vf$.
Then, since $\vf$ is a solution from $\cS_{W\cap U}$, there is $\bb\in\rel$
with $\ba=\pr_{W'}\bb$. Let $\bc=\bbf_\bs(\bb)$, clearly, $\bc\in\rel$.
For the tuple $\bc$ we have:\\[1mm]
-- $\bc[u]=f_u(\ba[u])=f_u(\vf(u))=\vf'(u)$ for $u\in W'$;\\[1mm]
-- $\bc[u]=f_u(\bb[u])=\psi(u)$ for $u\in\bs-W'$, because in this case
$f_u(\th_u)\sse\zz_u$, and therefore, as $f_u(\psi(u))=\psi(u)$, 
we have $f_u(\max(\zA_u))=\{\psi(u)\}$.\\[1mm]
Thus, $\bc=\pr_\bs\vf'$, and thus $\vf'$ is a solution from $\cS_{W\cap U}$.

So far we have defined mappings $\vf'_{w\gm\dl}$, proved that they are solutions
of the respective subinstances, that is, condition (1), and that they are 
consistent modulo $\ov\beta$, that is, condition (2). It remains to verify 
condition (3). Let $C=\ang{\bs,\rel}\in\cC$ and 
$\xi(u)=\vf_{w\gm\dl}(u)^{\beta_u}$, 
$\xi'(u)=\vf'_{w\gm\dl}(u)^{\beta_u}$ for 
$u\in V$ and any $(w,\gm,\dl)\in\cI$, such that $u\in W_{w\gm\dl}$. We need 
to show that $\pr_\bs\xi'\in\rel'=\rel\fac{\ov\beta_\bs}$.

We use a simplified version of the argument above. Let $W'=W\cap\bs$. If 
$v\not\in\bs$, the result follows from condition (3) for $\ov\gm$. Suppose 
$v\in W'$ and let $\bbf$ be a $(\beta_v,\gm_v)$-collapsing polynomial of 
$\rel'$. Also, let $\ba=\pr_\bs\xi$, 
$\bb'=\pr_{W\cap\bs}\vf\fac{\ov\beta_{W\cap\bs}}$, where 
$\vf=\vf_{v\al\beta}$ as before, and $\bb\in\rel'$ such that 
$\bb'=\pr_{W\cap\bs}\bb$. 
By Lemma~\ref{lem:relation-collapse} $\bbf$ can be selected such that 
$\ba\in\bbf(\rel')$ and $\bb[v]\in f_v(\zA_v\fac{\beta_v})$. Let 
$\bc=\bbf_{W\cap U}(\bb)$. We have\\[1mm]
-- $\bc[v]=\bb'[v]$;\\[1mm]
-- $\bc[u]=f_u(\bb'[u])=f'_u(\ba[u])=\ba[u]$ for $u\in W'-\{u\}$, as
$\vf(u)\in\xi(u)=\xi'(u)$;\\[1mm]
-- $\bc[u]=f_u(\bb[u])=f_u(\ba[u])=\ba[u]$ for $u\in\bs-W'$, as 
in this case $f_u(\th_u)\sse\beta_u$, and therefore, since 
$f_u(\ba[u])=\ba[u]$, we have 
$f_u(\max(\zA_u\fac{\beta_u}))=\{\bb[u]\}$.\\[1mm]
Therefore $\bc\in\rel'$, and as $\bc=\pr_\ba\xi'$, the result follows.
\end{proof}

\bibliographystyle{IEEEtran}
%% \bibliography{semilattice}

\end{document}